\documentclass[11pt]{article}
\usepackage[latin9]{inputenc}
\usepackage{geometry}
\usepackage{verbatim}
\usepackage{textcomp}
\usepackage{tipa}
\usepackage{amsmath}
\usepackage{amsthm}
\usepackage{amssymb}
\usepackage{graphicx}
\usepackage{setspace}
\usepackage[authoryear]{natbib}
\makeatletter

\InputIfFileExists{t2aenc.def}{}{%
  \errmessage{File `t2aenc.def' not found: Cyrillic script not supported}}
\DeclareRobustCommand{\cyrtext}{%
  \fontencoding{T2A}\selectfont\def\encodingdefault{T2A}}
\DeclareRobustCommand{\textcyr}[1]{\leavevmode{\cyrtext #1}}

\DeclareTextSymbolDefault{\textquotedbl}{T1}

\theoremstyle{definition}
\newtheorem{defn}{\protect\definitionname}
\theoremstyle{plain}
\newtheorem{lem}{\protect\lemmaname}
\theoremstyle{plain}
\newtheorem{prop}{\protect\propositionname}
\theoremstyle{remark}
\newtheorem{rem}{\protect\remarkname}
\theoremstyle{plain}
\newtheorem{assumption}{\protect\assumptionname}
\theoremstyle{plain}
\newtheorem{thm}{\protect\theoremname}
\theoremstyle{remark}
\newtheorem*{rem*}{\protect\remarkname}
\theoremstyle{plain}
\newtheorem*{cor*}{\protect\corollaryname}
\theoremstyle{plain}
\newtheorem*{lem*}{\protect\lemmaname}
\theoremstyle{plain}
\newtheorem*{prop*}{\protect\propositionname}
\theoremstyle{remark}
\newtheorem*{claim*}{\protect\claimname}

\AtBeginDocument{

}

\providecommand{\corollaryname}{Corollary}
\providecommand{\definitionname}{Definition}
\providecommand{\lemmaname}{Lemma}
\providecommand{\propositionname}{Proposition}

\usepackage{tikz}

\providecommand{\corollaryname}{Corollary}
\providecommand{\definitionname}{Definition}
\providecommand{\lemmaname}{Lemma}
\providecommand{\propositionname}{Proposition}

\providecommand{\corollaryname}{Corollary}
\providecommand{\definitionname}{Definition}
\providecommand{\lemmaname}{Lemma}
\providecommand{\propositionname}{Proposition}

\providecommand{\corollaryname}{Corollary}
\providecommand{\definitionname}{Definition}
\providecommand{\lemmaname}{Lemma}
\providecommand{\propositionname}{Proposition}

\usepackage[bottom]{footmisc}

\makeatother

\providecommand{\assumptionname}{Assumption}
\providecommand{\claimname}{Claim}
\providecommand{\corollaryname}{Corollary}
\providecommand{\definitionname}{Definition}
\providecommand{\lemmaname}{Lemma}
\providecommand{\propositionname}{Proposition}
\providecommand{\remarkname}{Remark}
\providecommand{\theoremname}{Theorem}

\begin{document}
\title{Policy with stochastic hysteresis }
\author{Georgii Riabov and Aleh Tsyvinski\thanks{Riabov: Institute of Mathematics, NAS of Ukraine; Tsyvinski: Yale
University. We thank Fernando Alvarez, Lint Barrage, Job Boerma, Jaroslav
Borovi\v{c}ka, Eduardo Faingold, Felix Kubler, Francesco Lippi, Ernest
Liu, Erzo Luttmer, Giuseppe Moscarini, Alessandro Pavan, Florian Scheuer,
Simon Scheidegger, Michael Sockin, Stefanie Stantcheva, Kjetil Storesletten,
Philipp Strack, Ted Temzelides, Pierre Yared, and Nicolas Werquin.}}
\maketitle
\begin{abstract}
The paper develops a general methodology for analyzing policies with
path-dependency (hysteresis) in stochastic models with forward looking
optimizing agents. Our main application is a macro-climate model with
a path-dependent climate externality. We derive in closed form the
dynamics of the optimal Pigouvian tax, that is, its drift and diffusion
coefficients. The dynamics of the present marginal damages is given
by the recently developed functional Itô formula. The dynamics of
the conditional expectation process of the future marginal damages
is given by a new total derivative formula that we prove. The total
derivative formula represents the evolution of the conditional expectation
process as a sum of the expected dynamics of hysteresis with respect
to time, a form of a time derivative, and the expected dynamics of
hysteresis with the shocks to the trajectory of the stochastic process,
a form of a stochastic derivative. We then generalize the results.
First, we propose a general class of hysteresis functionals that permits
significant tractability. Second, we characterize in closed form the
dynamics of the stochastic hysteresis elasticity that represents the
change in the whole optimal policy process with an introduction of
small hysteresis effects. Third, we determine the optimal policy process. 
\end{abstract}
\pagebreak{}

\section{Introduction}

Modern macroeconomics is stochastic and is built on recursive methods.
The actions in these models are history dependent but the past is
represented by a small number of finite-dimensional state variables.
For example, Ljungqvist and Sargent (2018) write that ``finding a
recursive way to handle history dependence is a major achievement
of the past 35 years and an important methodological theme''. They
call these developments an ``imperialistic response of dynamic programming''
to Prescott's (1977) critique of the impossibility of using dynamic
programming in government policy design. In this paper, we aim to
study environments in which it may in general be difficult or impossible
to write a recursive representation. We use the term hysteresis to
describe settings in which the trajectory of policies or actions may
generally affect the structure of the environment. That is, the effects
of policies are path-dependent.

Our main economic application is \textcyr{\cyra} generalization of
the macroeconomy-climate model of Golosov, Hassler, Krusell, and Tsyvinski
(2014) to the path-dependent climate externalities. There are two
main reasons to consider hysteresis in such models. First, there is
important recent evidence from climate sciences that a number of climate
variables show significant hysteresis behavior. Intergovernmental
Panel on Climate Change (IPCC) lists, for example, hysteresis that
is present in (1) vegetation change; (2) changes in the ice sheets;
(3) ocean acidification, deep ocean warming and associated sea level
rise; (4) models of the feedback between the ocean and the ice sheets;
and (5) models of the Atlantic meridional overturning circulation
(Collins, et al. 2013 in IPCC Fifth Assessment Report -- AR5). While
cumulative carbon emission is an important benchmark used in macro-climate
models, recent research in climate sciences (further discussed in
Section \ref{subsec:Evidence-on-climate}) points to importance of
hysteresis in a number of climate variables that constitute a natural
science foundation of these models. Moreover, the path-dependency
in these models is often rather sophisticated as the exact sequence
of the events may matter. The IPCC Special Report further concludes
that path dependence of carbon budgets ``remains an important knowledge
gap'' (Rogelj et al. 2018). Second, one of the central themes in
the recent developments in the economics literature on climate change
is incorporating path-dependency. A recent survey by Aghion, et al.
(2019) summarizes these developments and argues that path-dependency
in which both history and expectations matter is one of the core insights
in this literature (see, for example, Acemoglu, Aghion, Bursztyn,
and Hemous (2012), Acemoglu, Akcigit, Hanley, and Kerr (2016), Acemoglu,
Aghion, Barrage and Hemous (2019), and Aghion, Dechezlepretre, Hemous,
Martin, and Van Reenen (2016)). Finally, careful consideration of
uncertainty has been another important theme of recent research on
economics of climate change. Our application thus puts at the forefront
the analysis of an economy where \textcyr{\cyra} forward looking social
planner makes emission decisions under uncertainty and faces a path-dependent
emissions externality. In contrast to the literature and, specifically
to Golosov, Hassler, Krusell, and Tsyvinski (2014), we significantly
relax the assumption of how the previous emission choices enter the
climate damages. In our setting, the whole path of the emission (hysteresis)
matters as opposed to a depreciated stock of the previous emissions.
Specifically, we analyze a particular climate hysteresis functional
that captures the most important elements of the setting with general
hysteresis.

Our main result in this macro-climate application is a closed form
characterization of the dynamics of the marginal externality damage
which is also the Pigouvian tax that implements the optimal allocation.
We start by deriving the first-order conditions for the optimum and
showing that the marginal externality damage from emissions is comprised
of the marginal contemporaneous damages and the conditional expectation
of the cumulative future damages. Both of these terms are path-dependent
and their dynamics cannot be derived using the usual Itô formula which
applies only to functions of the current state and not the functionals
of the trajectories. We use two different sets of tools to provide
the dynamics of the present and the expected cumulative damages.

The present marginal damages are already represented as an explicit
path-dependent functional and we use recently developed functional
Itô calculus for the non-anticipative functionals (Dupire 2009, 2019,
Cont and Fournie 2013) to represent this term. The functional Itô
formula uses two new concepts of derivatives introduced by Dupire
(2009, 2019). The vertical derivative is an analogue of the space
derivative and the horizontal derivative is an analogue of the time
derivative for the functionals. These derivatives evaluate, respectively,
the effects of a discontinuous bump in the underlying trajectory and
the effects of a time extension of the trajectory while fixing the
terminal position. The Dupire's functional Itô formula then gives
the dynamics of the present damages with the drift determined by the
horizontal and the second vertical derivative, and the diffusion coefficient
given by the vertical derivative of the marginal present damages.
Using the notion of the horizontal and vertical derivatives the functional
Itô formula allows us to convey the similar intuition as in the case
of no hysteresis for which the usual Itô formula is applicable. However,
in contrast with the case of no hysteresis path-dependency may lead
to significant effects on the dynamics even in the case of absent
contemporaneous damages.

The most challenging part of the paper is to characterize the dynamics
of the expected future marginal damages represented by the conditional
expectation process. The difficulty comes from the fact that with
passage of time both the damages themselves and the information set
(filtration) are changing. The main new theoretical tool that we develop
-- the total derivative formula for conditional expectation processes
-- allows us to characterize the dynamics of such processes. Moreover,
this formula is broadly applicable to a variety of macroeconomic settings
where the conditional expectations processes are central. There are
two terms in the semimartingale decomposition of the conditional expectations
process that the total derivative formula delivers. The first term
can be intuitively thought of as \textcyr{\cyra} time derivative and
represents how the conditional expectation of the future marginal
damages evolves with respect to time. The second term, given as a
conditional expectation of the Malliavin derivative of the cumulative
future marginal damages, can be thought of as a stochastic derivative
with respect to the underlying process and represents how the conditional
expectation changes with the changes in the stochastic process. It
is useful to compare the results to the case with no hysteresis. The
main difference is that in the case of no hysteresis, the policy does
not have future effects. With hysteresis, the contemporaneous effects
of the policy may be very small but the expected future effects of
the policy may be very large -- small actions today my have significant
future consequences. Moreover, these effects change with both time
and the stochastic shocks.

Combining the dynamics of the present and the conditional expectation
of the future marginal damages, we obtain in closed form the drift
and diffusion of the optimal Pigouvian tax correcting the path-dependent
externality. The dynamics are given by the semimartingale decomposition
of the marginal externality damages which can be though of in terms
of the planner changing the tax with the passage of time and with
the realization of uncertainty. Alternatively, this can be thought
of as determination of what features of the model matter to the first
order for the drift and diffusion of the optimal tax. Specifically,
those are given by the functional Itô formula and by our total derivative
formula for the conditional expectation processes. Even if there are
no direct contemporaneous effects, the dynamics of the optimal policy
may be significantly impacted by either the past choices or by the
future effects. In other words, both the past and the future non-trivially
matter.

We then show that one can significantly generalize the application
of the climate hysteresis. First, we propose a very general class
of hysteresis functionals that allow significant tractability of the
analysis of optimal policies. The Fréchet derivatives of these functionals
have two sources of variation: an instantaneous influence in the given
(current) period and an integral influence of perturbations in previous
periods. The assumption is mild and only requires that the derivative
of a functional at a given time is absolutely continuous up to that
period and may have an atom at the present time. We show that every
Fréchet differentiable functional is a pointwise limit of functionals
from this class. Importantly, these functionals have an attractive
property of conveniently separating the effects of the past and the
present. As we are considering a very general environment with path-dependence,
it may be difficult or impossible to use dynamic programming to write
a recursive representation of the problem. Instead, we directly find
the first order conditions for the optimal problem by considering
all local variations of the policies. The assumption on the derivative
of the hysteresis functional allows us to conveniently separate the
first order condition in two parts: the marginal effects of the policy
on current period and the conditional expectation of the cumulative
future marginal effects of policy paths in the future periods.

Our next goal is to explicitly find the evolution equation for the
stochastic hysteresis elasticity. This concept of stochastic hysteresis
elasticity is similar to the usual concept of elasticity as it captures
the change in the variables following a small change in policy. The
difference is that the stochastic elasticity is the change in the
whole optimal process. Alternatively, one can think of the stochastic
elasticity as the asymptotics of the optimal policy when there is
an infinitesimal additional hysteresis. Specifically, we show that
the stochastic elasticity is an Itô process and write its explicit
semimartingale decomposition. A fundamental property of the solution
is that both of the terms -- the present and the expected future
marginal effects of policy paths -- in the first order condition
for optimality and in the formula for the stochastic elasticity are
path-dependent. Similarly to our climate hysteresis application we
use the functional Itô formula and the total derivative formula to
represent the dynamics of the stochastic elasticity. The semimartingale
decomposition of the present and the cumulative marginal effects of
policies show that the stochastic elasticity, that is, the change
in the optimal policy process following introduction of hysteresis,
is an Itô process. The difference with the no-hysteresis case is that
now the drift and diffusion coefficients are themselves functionals
of the trajectory of the policy. That is, hysteresis leads to the
path-dependent stochastic process of the elasticity for which we calculate
the drift and diffusion coefficients in the explicit form. Finally,
we return to the analysis of the optimal policies with general hysteresis.
The evolution of the optimal policy can be analyzed using exactly
the same tools as the stochastic elasticity. The explicit form we
obtained for the stochastic elasticity thus can be thought of as a
characterization of the optimal policy process for small hysteresis,
or small hysteresis asymptotics. We use the functional Itô formula
to represent the current marginal effects and the total derivative
formula to represent the conditional expectation process of the future
marginal effects of policies. The drift and diffusion coefficients
of the optimal policy are then determined implicitly by the equations
having a similar form as the explicit solution for the stochastic
elasticity but also depending on the process of the optimal policy
itself.

We finally provide several examples of how to calculate the stochastic
hysteresis elasticity. All of these examples are straightforward applications
of the main semimartingale decomposition formula that we developed. 

\subsection*{Literature}

The closest to our work is a sequence of papers by Borovi\v{c}ka,
Hansen, and Scheinkman (2014), Borovi\v{c}ka, Hansen, Hendricks, and
Scheinkman (2011), and Borovi\v{c}ka and Hansen (2016). They define
the concept of the shock elasticity, relate it to impulse responses
familiar to macroeconomists, and, importantly, formalize it using
Malliavin derivatives. The shock elasticity is an impulse response
of the pricing kernel to the marginal perturbation of the underlying
stochastic cash flow process. Our work goes beyond their results and
delivers the explicit semimartingale decomposition of the evolution
of the stochastic elasticity for the path-dependent functionals. In
Section \ref{subsec:Relationship-to-the_Clarrk-Ocone} we also provide
a connection to the Clark-Ocone formula that features prominently
in Borovi\v{c}ka, Hansen, and Scheinkman (2014).

The only paper that we are aware of that uses functional Itô formula
in economic settings is an insightful work of Cvitani\'{c}, Possamaï,
and Touzi (2017). In the dynamic moral hazard setting where the agent
controls the volatility of the process, they show how that without
loss of generality the set of admissible contracts can be represented
as the path-dependent processes arising from the functional Itô formula.
In our environment, both the functional Itô formula and the total
derivative formula are applied. We discuss this in more details in
Section \ref{subsec:Why-do-we-need-two-formulas}.

The concept of stochastic elasticity is similar to the analysis of
the models of investment under uncertainty for the case of small sunk
costs or for models with small menu costs. In those models, there
is a form of hysteresis where due to the region of inaction, a temporary
change may have a permanent effect. The analysis of such models have
important implications for the large amounts of inertia that arise
in the models of small menu costs such as in Akerlof and Yellen (1985)
and Mankiw (1985). Dixit (1991) develops a method of analytical approximation
for such hysteresis models. Reis (2006) and Alvarez, Lippi, and Paciello
(2011, 2016) derive a similar analytical characterization in the models
with inattentive producers. These are essentially the same as our
calculation of stochastic elasticity for the environments that can
be characterized by Itô formula. A series of papers by Alvarez and
Lippi (2019), Alvarez, Le Bihan, and Lippi (2016) and Alvarez, Lippi,
and Oskolkov (2020) focus on the cumulative output effects of shock
in environments with the menu and adjustment costs frictions. These
cumulative output effects are the output impulse response functions\footnote{Similarly, the impulse responses play an important role in Pavan,
Segal, and Toikka (2014), Garrett and Pavan (2012), and Bergemann
and Strack (2015) in dynamic mechanism design, Makris and Pavan (2017)
in the dynamic taxation, and Makris and Pavan (2018) who connect these
two literatures. See also a review in Pavan (2017). In general, the
recursive formulation of a number of these problems may be challenging
and the tools we develop may prove to be useful in those contexts.} that are similar to the future expected marginal effects of policy
in our paper. Our analysis can be thought of as generalizing the results
to the environments in which there is path-dependence.

An important early paper by Detemple and Zapatero (1991) studies an
asset pricing model with habit formation using Malliavin calculus.\footnote{See also Serrat (2001), Bhamra and Uppal (2009) and Cvitanic and Malamud
(2009). } Huang and Sockin (2018) study optimal dynamic taxation in continuous
time and calculate impulse response to misreporting using Malliavin
derivatives. As Sannikov (2014) for the model of moral hazard they
connect their analysis to the finance literature on the sensitivity
of the ``Greeks'' of the options (see, e.g., Fournié, Lasry, Lebuchoux,
Lions, and Touzi 1999). Their analysis relates to our calculation
of Malliavin derivatives of the policy process in the last part of
Section \ref{subsec:Example Martingale}.

Hysteresis or path-dependence is important in a variety of macroeconomic
and other settings. The natural rate of unemployment may directly
depend on the previous path of unemployment and, hence, on the path
of monetary policy (Blanchard and Summers 1986). Examples of recent
explicit focus on path dependency and monetary policy is Berger, Milbradt,
Tourre, and Vavra (2018), Gali (2020), and Jorda, Singh, and Taylor
(2020). Hysteresis is also central in the models of secular stagnation
as in Summers (2014) and Eggertsson, Mehrotra, and Robbins (2019).
Adjustment costs in investment generate hysteresis (for example, see
a survey in Dixit 1992). Acemoglu, Egorov, and Sonin (2020) exposit
a wide variety of political economy models with path-dependence of
policies or institutions. Egorov and Sonin (2020) additionally discuss
models in which there is non-Markovian dependence on the sequence
of policy events. Page (2006) classifies types of history dependence
and distinguishes between the path-dependent policies (where the exact
sequence of events matters) and the ``phat'' dependent strategies
(where the events but not their sequence matters). The models of increasing
returns such as Arthur (1989) and David (1985), and a review in Arrow
(2000) feature potentially very complicated dependence on the paths.

We now briefly discuss the relationship of the economic environment
we consider with mathematical literature. The analysis of deterministic
dynamic systems with hysteresis is well established (e.g., Krasnosel'skii
and Pokrovskii 2013). In our model, path dependency arises from the
optimization of the forward looking agents under uncertainty rather
than being directly posed as a property of a dynamic system. The problem
that we study also does not fit easily into the standard framework
of the infinite-dimensional stochastic optimal control such as Fabbri,
Gozzi, and Swiech (2017), see discussion in Section \ref{subsec:Discussion-of-the-optimal control}.
The reason is that the hysteresis functional depends on the whole
path of the policy. If we treat the policy as a control process we
are obliged to take the space of paths (of variable length) as a control
space which leads to the analysis of path-dependent HJB equations
the understanding of which is in its nascency.

\section{Environment}

We describe an environment of a stochastic problem with forward looking
agents and path-dependency.

\subsection{\label{subsec:Baseline-setting}Baseline setting}

Let $w=(w_{t})_{0\leq t\leq T}$ be a Brownian motion.\footnote{All of the results can be straightforwardly extended to more general
diffusion processes.} Let the objective of the planner be described by the quadratic loss
function:
\begin{equation}
\max_{c}\mathbb{E}\int_{0}^{T}-\frac{1}{2}\left(c_{t}-w_{t}\right)^{2}dt.\label{eq:Unperturbed problem}
\end{equation}
Here maximization is taken over all (progressively measurable) policy
processes $c:[0,T]\times C([0,T])\to\mathbb{R}.$ That is, policy
$c_{t}$ depends on information about $w$ up to moment $t$. 

The basic problem is easily solved by interchanging integral and expectation
and yields the optimal process $c_{t}^{*}$ : 
\begin{equation}
c_{t}^{*}=w_{t}.\label{eq: Optimal consumption, unperturbed}
\end{equation}
The optimal policy simply tracks the Brownian motion $w_{t}$.

\subsection{General environment}

Consider a problem where there is an additional path-dependent effect
of policy $h_{t}$: 
\begin{equation}
\max_{c}\mathbb{E}\int_{0}^{T}(-\frac{1}{2}\left(c_{t}-w_{t}\right)^{2}-\epsilon h_{t}(c))dt,\label{eq:Perturbed problem}
\end{equation}
where $h:[0,T]\times C[0,T]\to\mathbb{R}$ is an adapted functional
on the space of trajectories (i.e. $h_{t}(c)$ is defined only by
the restriction of $c$ on $[0,t]$ ), and $\epsilon$ is a parameter.
One can think of the functional $h$ as an additional effect (a cost
or a benefit) that depends on the history of policy up to time $t$.
The solution of this problem is denoted $c^{\epsilon}$. We define
hysteresis as dependence of the effects of policy $h_{t}(c)$ on the
path of the previous policies $c_{[0,t]}=c^{t}$ .

\section{Climate externality}

This section generalizes the macro-climate model of Golosov, et al.
(2014) to the case where climate externality is path-dependent. Section
\ref{subsec:Planner's-problem-Climate} shows how to map some of the
key features of that paper to the framework described in the previous
section. Section \ref{subsec:Evidence-on-climate} discusses evidence
on hysteresis in the climate externality. Section \ref{subsec:Climate-hysteresis-functional}
proposes a specific path-dependent climate hysteresis functional that
captures most of the insights of the general framework that we develop
later in Section \ref{sec:General-setting}. Sections \ref{subsec:First-order-conditions-climate}-\ref{subsec:Dynamics Pigouvian Climate}
solve in closed form the dynamics of the marginal climate externality
or the Pigouvian tax implementing the optimum, that is, determine
its drift and diffusion coefficients.

\subsection{\label{subsec:Planner's-problem-Climate}Planner's problem with a
climate externality}

We first note that the structure of the problem (\ref{eq:Perturbed problem})
captures some of the main features of the macroeconomic model with
a climate externality of Golosov, et al. (2014). In its essence, a
macro-climate model consists of two primary blocks. The first is the
specification of the economy which is a standard dynamic stochastic
general equilibrium model. The planner maximizes the expected utility
of consumption (we assume no discounting)
\[
\mathbb{E}\int_{0}^{T}U\left(c_{t}^{\text{cons}}\right)dt,
\]
subject to the standard feasibility constraint with capital accumulation,
where $c_{t}^{\text{cons}}$ is a consumption good. The production
function is given by 
\[
Y_{t}=F\left(K_{t},E_{t},S_{t}\right),
\]
where $K_{t}$ is capital, $E_{t}$ is energy consumption (measured
in carbon emission units), and $S_{t}$ is a climate variable at time
$t$. The climate variable is in general a functional $S_{t}=\tilde{S_{t}}\left(E^{t}\right)$
of the path of emissions $E^{t}=E_{\left[0,t\right]}$ (equation (4)
in Golosov, et al. (2014) is, in fact, exactly this general formulation).
The climate variable $S_{t}$ affects the economy via a damage function
$D_{t}\left(S_{t}\right)$ so that the output is reduced multiplicatively
$\left(1-D_{t}\left(S_{t}\right)\right)\times\tilde{F}\left(K_{t},E_{t}\right).$
Moreover, the damages are assumed to be exponential $1-D_{t}\left(S_{t}\right)=\exp\left(-\gamma_{t}S_{t}\right)$
for some parameter $\gamma_{t}$ and the utility $U\left(x\right)=\log\left(x\right)$.
The essence of the model is how the path of emissions $E^{t}$ translates
to damages to the economy. Let us now reduce the model further to
focus on this key dimension. Assume that production does not require
capital and the feasibility constraint for the economy is then static:
\[
c_{t}^{cons}=Y_{t}=e^{-\gamma_{t}S_{t}}F\left(E_{t}\right).
\]
The planner then maximizes
\[
\max_{E_{t}}\mathbb{E}\int_{0}^{T}\left(\tilde{U}\left(E_{t}\right)-\gamma_{t}S_{t}\right)dt,
\]
which is essentially equivalent to the general problem (\ref{eq:Perturbed problem}).
We further assume that $\tilde{U}\left(E_{t}\right)=-\frac{1}{2}\left(E_{t}-w_{t}\right)^{2}$
but this is again done purely for leanness of the model and can be
immediately extended. 

The key simplification, however, that the literature makes is in simple
dependence of the climate variable $S_{t}$ on the path of the emissions.
For example, Golosov, et al. (2014)) assume that 
\[
S_{t}=\int_{0}^{t}d_{s}E_{s}ds,
\]
where $d_{s}\in\left[0,1\right]$ is a carbon depreciation rate and
then also further structure is placed on $d_{s}$. In other words,
the climate variable is equal to the stock of the depreciated emissions.

The climate-economy model in its essence reduces to how the previous
energy consumption choices and the associated carbon emissions impact
today's and future economy. In contrast to the literature, we significantly
relax the assumption of how the previous emission choices enter the
damages. In our setting, the whole path of the emissions (hysteresis)
matters through the functional $h_{t}\left(c^{t}\right)$ as opposed
to a depreciated stock of the previous emissions. We, on purpose,
stripped down the climate model to focus only on the general path-dependent
effects of the emissions and their interaction with the uncertainty.

\subsection{\label{subsec:Evidence-on-climate}Evidence on climate hysteresis}

There are two primary sets of evidence which imply that hysteresis
is important in the economic models of climate change.

The first is a large set of recent evidence on importance of hysteresis
in climate sciences. The most comprehensive and authoritative source
for such research is the Intergovernmental Panel on Climate Change
(IPCC). The working group for the physical science basis of the long-term
climate change in the Fifth Assessment Report -- AR5 (Collins, et
al. 2013), while arguing for the attractiveness of the use of the
cumulative carbon emission, notes that a number of climate variables
and models show significant hysteresis behavior. These are the models
of (1) vegetation change; (2) changes in the ice sheets; (3) ocean
acidification, deep ocean warming and associated sea level rise; (4)
models of the feedback between the ocean and the ice sheets; and (5)
models of the Atlantic meridional overturning circulation. The report
also argues that ``the concepts of climate stabilization and targets
is that stabilization of global temperature does not imply stabilization
for all aspects of the climate system. For example, some models show
significant hysteresis behaviour in the global water cycle, because
global precipitation depends on both atmospheric CO2 and temperature''.
This is also consistent with the study of Zickfield et al. (2012)
who argue that, while total cumulative emissions may be an important
approximation for many models, there are a number of exceptions where
path-dependency is important: among the variables with timescales
of several centuries, such as deep ocean temperature and sea level
rise, and for the peak responses of atmospheric CO2 or for the surface
ocean acidity. We now briefly discuss some other recent evidence of
the hysteresis behavior in individual components of the climate system.
Eliseev et al. (2014) is a comprehensive study of the permafrost behavior
that finds significant evidence of hysteresis especially for the higher
concentration of greenhouse gases in the atmosphere. Another important
aspect of ice thawing is the release of permafrost carbon which is
found to be both very significant and highly path-dependent (Gasser
et al. 2018). Garbe, et al. (2020) analysis ``reveals a strong, multi-step
hysteresis behaviour of the Antarctic Ice Sheet'' which is potentially
reinforced by a number of additional feedback mechanisms. Nordhaus
(2019) augments the DICE model with the effects of Greenland ice sheet
disintegration and finds that the baseline or the no-policy effects
are significantly different when even a simple formulation of the
hysteresis is incorporated while the optimal policy results are similar.
Boucher et al. (2012) analyzes the response of a number of models
to a significant increase in the CO2 concentration and finds that
hysteresis is particularly pronounced for the terrestrial (such as
long lived soil carbon sinks and vegetation) and marine variables
(such as the sea-level rise) as well as in global mean precipitation.
Nohara et al. (2013) finds that the hysteresis effects are important
for the regional climate change which is notable given, for example,
the recent focus on the significant differences in optimal climate
policy at the regional level by Hassler, et al. (2020). Wu et al.
(2012) describe significant hysteresis in the hydrological cycle that
leads to one of the most direct impacts of global warming affecting
droughts, floods and water supplies. This is important for the recent
economic literature that studies the impact of the floods and the
sea level (for example, Bakkensen and Barrage 2017, Barrage and Furst
2019; Hong, Wang, and Young, 2020 where hysteresis may be also compounded
with the path-dependency in beliefs). Summarizing, while the cumulative
emission and its variants are certainly a central benchmark for the
development of the climate-economy models, recent research in climate
sciences points to importance of hysteresis with sophisticated trajectory
dependence in a number of important climate variables.

The second reason for why hysteresis is important lies in the economics
part of the climate-economy models. One of the central themes in the
recent advances in this literature has been incorporating path-dependency.
A recent survey by Aghion, et al. (2019) summarize these developments:
``The core insight is that technological innovation is a path-dependent
process in which history and expectations matter greatly in determining
eventual outcome'' leading to ``important implications for climate
policy design'' and ``research and knowledge production are path-dependent,
deployment of innovations is path-dependent and the incentives for
technology adoption create path dependence''. Several papers in the
literature develop various parts of these insights. In Acemoglu, Aghion,
Bursztyn, and Hemous (2012), dirty technologies have an advantage
in the market size and in the initial productivity and, hence, there
is path-dependency in the direction of innovation and production.
One of the most important findings of Acemoglu, Akcigit, Hanley, and
Kerr (2016) is that the nature of innovations in clean or dirty technologies
is path dependent. Similarly, there is path-dependence in innovation
in the model of the consequences of the shale gas revolution by Acemoglu,
Aghion, Barrage and Hemous (2019). Grubb, et al. (2020) and Baldwin,
Cai, and Kuralbayeva (2020) develop models of path-dependency due
to the costs of switching from the dirty technologies. Aghion, Dechezleprêtre,
Hemous, Martin, and Van Reenen (2016) provide extensive empirical
evidence of path dependence in innovation in auto industry from aggregate
spillovers and from the firm\textquoteright s own innovation history.
Meng (2016) finds significant evidence for strong path dependence
in energy transition for the U.S. electricity sector over the 20th
century, focusing on coal. Fouquet (2016) is a summary of evidence
that energy systems are subject to strong and long-lived path dependence
due to technological, infrastructural, institutional and behavioral
lock-ins. A recent strand of the climate economics literature (Lemoine
and Traeger 2014; Lontzek, Cai, Judd, and Lenton 2015; van der Ploeg
and de Zeeuw 2018, and Cai and Lontzek 2019) has attempted to incorporate
hysteretic effects through the tipping point modeling. Dietz, Rising,
Stoerk, and Wagner (2020) synthesize a number of different approaches
in this literature into a meta-model and argue for the need to develop
sophisticated models of hysteretic behavior.\footnote{Section \ref{subsec:A-tipping-point} further elaborates on an example
of the use of our methodology with tipping points.}

Finally, careful consideration of uncertainty such as in Temzelides
(2016), Li, Narajabad, and Temzelides (2016), Traeger (2017), Cai
and Lontzek (2019), Brock and Hansen (2018), Van den Bremer and van
der Ploeg (2018), Barnett, Brock, and Hansen (2020), Giglio, Kelly,
and Stroebel (2020), Kotlikoff, Kubler, Polbin, and Scheidegger (2020)
and Lemoine (2021) has been another important theme of recent research
on economics of climate change. 

\subsection{\label{subsec:Climate-hysteresis-functional}Climate hysteresis functional}

In this section, we analyze a particular hysteresis functional that
captures the most important elements of the abstract setting that
we develop in Section \ref{sec:General-setting} and allows to transparently
show how the solution in the general environment works.

Specifically, let 
\begin{equation}
h_{t}\left(c^{t}\right)=g_{t}\left(w^{t}\right)c_{t}+\int_{0}^{t}k_{s,t}\left(w^{t}\right)c_{s}ds,\label{eq: Climate hysteresis functional}
\end{equation}
where $c_{t}$ denotes the amount of emissions at time $t$.\footnote{Strictly speaking, the functional in this section depends on both
the path $c^{t}$ and $w^{t}$ while we consider in Section \ref{sec:General-setting}
the functional $h_{t}\left(c^{t}\right)$. The extension to allow
for the dependence on the path $w^{t}$ is immediate. Further, for
the analysis of the stochastic elasticity in Section \ref{sec:First-order-process},
the underlying optimal process $c_{t}^{*}=w_{t}$ and, hence, this
is without loss of generality. } Here, there are two sources of the effects of emissions. The first
source is contemporaneous. The amount of emission $c_{t}$ yields
climate externality equal to $g_{t}\left(w^{t}\right)$, where $g_{t}$
is a functional of the path of uncertainty $w^{t}$. The second source
is the effect of the past emissions $c_{s}$. Each of the emissions
in the previous period $c_{s}$ contributes $k_{s,t}\left(w^{t}\right)$
to the climate externality $h_{t}$ at time $t$, where $k_{s,t}$
is a functional of the path of uncertainty $w^{t}$. It is important
to note that both of these terms are path-dependent.\footnote{Remark \ref{rem:Integration by parts} in Section \ref{subsec:Class A_t}
further discusses motivation for this example as a representation
of stochastic integrals.}

We observe that with the functional (\ref{eq: Climate hysteresis functional})
the optimization problem (\ref{eq:Perturbed problem}) becomes intractable
for the stochastic control approach. Indeed, two possibilities for
such approach would be either to introduce a control $c^{t}$ or to
introduce an additional control $x_{t}=h_{t}(c_{t})$. Both approaches
lead to complicated infinite-dimensional HJB equations. The first
is due to the complicated structure of the control space. The second
is due to the path dependent first order condition for the process
$x.$\footnote{Of course, for some special cases there is no need to use the theory
we develop. In Section \ref{sec:Examples-of-perturbation} we show,
whenever possible, how to derive the results using known tools. One
interesting class of examples can be solved by extracting a martingale
and then using the Clark-Ocone formula. In Section \ref{subsec:Example Cumulative-hysteresis}
and Remark \ref{rem:Boulatov} we show how to do this for the case
of cumulative hysteresis $h_{t}\left(c_{t}\right)=c_{t}\int_{0}^{t}c_{s}ds$;
in Section \ref{subsec:Example Martingale} and Remark \ref{rem:Detemple Zapatero}
we show how to do this for when the kernel $k_{s,t}$ is mutliplicative.
We further expand on this class of examples in Boulatov, Riabov, and
Tsyvinski (2020).}

\subsection{\label{subsec:First-order-conditions-climate}First-order conditions}

We derive the first order conditions of the problem (\ref{eq:Perturbed problem})
with the climate hysteresis functional (\ref{eq: Climate hysteresis functional})
by perturbing the process $c$ by a an adapted process $z$ and computing
the derivative in $\nu$ at $\nu=0:$
\begin{align*}
 & \partial_{\nu}E\int_{0}^{T}\left(-\frac{1}{2}(c_{t}+\nu z_{t}-w_{t})^{2}-\epsilon\left(g_{t}\left(w^{t}\right)\left(c_{t}+\nu z_{t}\right)+\int_{0}^{t}k_{s,t}\left(w^{t}\right)\left(c_{s}+\nu z_{s}\right)ds\right)\right)dt\bigg|_{\nu=0}=\\
 & =E\int_{0}^{T}\left(z_{t}w_{t}-c_{t}z_{t}-\epsilon\left(g_{t}\left(w^{t}\right)z_{t}+\int_{0}^{t}k_{s,t}\left(w^{t}\right)z_{s}ds\right)\right)dt=\\
 & =E\int_{0}^{T}z_{t}\left(w_{t}-c_{t}-\epsilon\left(g_{t}\left(w^{t}\right)+\int_{t}^{T}k_{t,s}\left(w^{s}\right)ds\right)\right)dt=0,
\end{align*}
where the last line is obtained by changing the order of integration.
Since $z$ is an arbitrary adapted process we get the first order
conditions
\begin{equation}
c_{t}=w_{t}-\epsilon\left(g_{t}\left(w^{t}\right)+E\left[\int_{t}^{T}k_{t,s}(w^{s})ds\bigg|\mathcal{F}_{t}\right]\right).\label{eq:FOC emissions}
\end{equation}
The term
\begin{equation}
\Lambda_{t}=g_{t}\left(w^{t}\right)+E\left[\int_{t}^{T}k_{t,s}(w^{s})ds\bigg|\mathcal{F}_{t}\right]\label{eq: Pigouivan tax emissions}
\end{equation}
plays the central role in Golosov, et al. (2014) and is the marginal
externality damage from emissions, it also is equal to the optimal
Pigouvian tax that corrects this externality. Specifically, in our
setting there are two effects of the marginal change in the emission
$c_{t}$. First, there is the immediate damage $g_{t}\left(w^{t}\right)$.
Second, there are marginal damages in each future period $s$, given
by $k_{t,s}\left(w^{s}\right).$ These damages are then integrated
and evaluated as the expectation conditional on the information up
to time $t$.

It is important to note the difference with the case with no hysteresis,
where the externality is a function $g_{t}\left(w_{t}\right)$ and
not a functional of the trajectory $w^{t}$. First, the future expected
marginal damage $E\left[\int_{t}^{T}k_{t,s}(w^{s})ds\bigg|\mathcal{F}_{t}\right]$
is absent in that case as the change in emissions today does not have
effects on the future. Second, the contemporaneous effects $g_{t}\left(w^{t}\right)$
also may be significantly different from $g_{t}\left(w_{t}\right)$
as it depends now on the path $w^{t}$. Even if emissions have no
immediate effects, the past effects and the expected cumulative effects
may be very significant. In other words, both the past and the future
non-trivially matter.

Our primary interest is in providing the evolution of the marginal
externality damage, that is, in the semimartingale decomposition $d\Lambda_{t}=\alpha_{t}\left(w^{t}\right)dt+\beta_{t}\left(w^{t}\right)dw_{t}$.
This decomposition can be thought of as the evolution of the marginal
externality damage (or the optimal Pigouvian tax) with respect to
time $dt$, the drift, and with respect to the realizations of uncertainty
$dw_{t}$, the diffusion coefficient. Since both the terms $g_{t}\left(w^{t}\right)$
and $E\left[\int_{t}^{T}k_{t,s}(w^{s})ds\bigg|\mathcal{F}_{t}\right]$
are path-dependent, we cannot use the standard Itô formula as it applies
only to the functions and not to the functionals of the path. We proceed
by analyzing the present and the future effects separately as the
analysis requires the use of different methods.

\subsection{Present effects of the emissions path\label{subsec:Present-effects-Example}}

The present damage $g_{t}\left(w^{t}\right)$ of the current emissions
is already a well-defined functional of the path $w^{t}$. We use
a recently developed Dupire's functional Itô formula for the non-anticipative
functional (Dupire 2009, 2019; Cont and Fournie 2013) to represent
this term.

We first introduce two derivatives due to Dupire (2009, 2019): the
vertical and the horizontal derivative. These are, respectively, the
functional analogues of the space and of the time derivatives.
\begin{defn}
\textbf{(Vertical Derivative)} Let $D[0,t]$ be the space of càdlàg
paths. The vertical derivative of a functional $h:D[0,t]\to\mathbb{R}$
is the limit
\[
\partial_{c_{t}}h(c)=\lim_{\epsilon\to0}\frac{h(c+\epsilon e_{t})-h(c)}{\epsilon},
\]
where $e_{t}(s)=1_{s=t}.$
\end{defn}
Intuitively, this derivative measures an influence on the functional
of a (discontinuous) bump to the path at time $t$ (Figure \ref{fig:Dupire's-vertical-derivative}).
If the functional is just the function of current realization of the
path, this yields the usual derivative $h'\left(c_{t}\right).$

\begin{figure}[h]
\centering{}\includegraphics[scale=3]{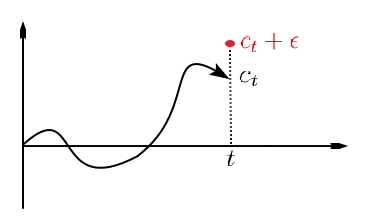}\caption{\label{fig:Dupire's-vertical-derivative}Dupire's vertical derivative:
path perturbation}
\end{figure}

\begin{defn}
\textbf{(Horizontal derivative)} Let $z$ $:[0,T]\times D[0,T]\to\mathbb{R}$
be an adapted functional on the space of trajectories. The horizontal
derivative $\Delta_{t}$ is defined as a limit 
\[
\Delta_{t}z_{t}(c)=\lim_{\epsilon\to0}\frac{z_{t+\epsilon}(c_{\cdot,\epsilon})-z_{t}(c)}{\epsilon},
\]
where $c_{\cdot,\epsilon}$ is an extension of the path $c$ from
$[0,t]$ to $[0,t+\epsilon]$ by $c_{s,\epsilon}=c_{t}$ for $t\leq s\leq t+\epsilon.$
\end{defn}
Intuitively, this derivative freezes the path of the underlying process
at time $t$, extends it for a small time $\epsilon$ and evaluates
the functional $z_{t+\epsilon}$ on this extended path (Figure \ref{fig:Dupire's-horizontal-derivative}).
If the functional $z_{t}$ is just the function of the value of the
path at time $t$, this yields the usual time derivative of $(\partial_{t}z_{t})(c_{t}).$

\begin{figure}[h]
\centering{}\includegraphics[scale=3]{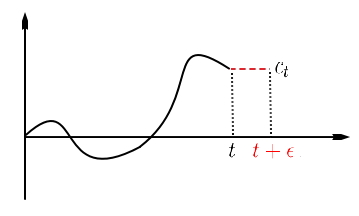}\caption{\label{fig:Dupire's-horizontal-derivative}Dupire's horizontal derivative:
path perturbation}
\end{figure}

If the functional $g_{t}\left(w^{t}\right)$ does not depend on the
path of policies and is a function $f\left(w_{t},t\right)$ of only
the current realization of uncertainty and of time, we can use the
usual Itô formula. For its application, we need to ensure that the
first and the second order space derivative $f'_{1}\left(w_{t},t\right)$,
$f''_{1}\left(w_{t},t\right)$ and the time derivative $f'_{2}\left(w_{t},t\right)$
exist. When, as is the case in this section, there is hysteresis to
apply the functional Itô formula we need to make a similar assumption
on the functional $g_{t}\left(w^{t}\right)$ as in the case of it
being a function but with different notions of derivatives.

\begin{lem}
\textbf{(Functional Itô Formula: Dupire, 2009, 2019; Cont and Fournie,
2013)\label{lem:Functional-It=0000F4's-Formula:}} Let $x_{t}$ be
an Itô process, $dx_{t}=b_{t}dt+\sigma_{t}dw_{t},$ where the drift
coefficient $b_{t}$ and the diffusion coefficient $\sigma_{t}$ are
adapted functionals of $w$. Let $g_{t}$ be horizontally differentiable
and twice vertically differentiable functional on the space of cádlág
paths. Then the process $g_{t}(x^{t})$ is an Itô process with the
stochastic differential
\begin{equation}
dg_{t}\left(x^{t}\right)=\Delta_{t}g_{t}(x^{t})dt+\partial_{c_{t}}g_{t}(x^{t})dx_{t}+\frac{1}{2}\partial_{c_{t}}^{2}g_{t}(x^{t})(\sigma_{t})^{2}dt.\label{eq:Functional Ito}
\end{equation}
\end{lem}
The lemma gives a semimartingale decomposition of the functional $g_{t}$$\left(x^{t}\right)$.
It is similar in form to the usual Itô formula that applies to the
functions of the state. However, the time derivative of a function
is replaced by the horizontal derivative of the functional $\Delta_{t}$;
the first and second space derivatives of the function are replaced
by the first and second vertical derivatives $\partial_{c_{t}}$ and
$\partial_{c_{t}}^{2}$. Of course, if there is no path-dependency
and $g_{t}$ is a function rather than a functional, one recovers
the usual Itô formula.

Applying functional Itô formula to $g_{t}\left(w^{t}\right)$ we obtain\textbf{
\[
dg_{t}\left(w^{t}\right)=\left(\Delta_{t}g_{t}(w^{t})+\frac{1}{2}\partial_{c_{t}}^{2}g_{t}(w^{t})\right)dt+\partial_{c_{t}}g_{t}(w^{t})dw_{t}.
\]
}

Summarizing the results in this section: the drift of the dynamics
of the present effects of the emissions is determined by the horizontal
and the second vertical derivative and the diffusion coefficient is
given by the vertical derivative of the present marginal effects of
the policy functional. This result both parallels and differs from
the case of no hysteresis where, similarly, the first and second order
derivatives matter for the dynamics. However, the notion of the derivative
is very different as the horizontal and the vertical derivatives measure
how hysteresis functionals change with the whole past trajectory.
Even if there are no direct contemporaneous effects of the policy,
its dependence on the past may lead to significant effects on the
dynamics of optimal policy. At the same time, using these different
notions of the derivatives the functional Itô formula allows us to
convey the similar intuition as in the case of no hysteresis. In particular,
uncertainty represents itself in the second vertical derivative of
the functional $g_{t}\left(w^{t}\right)$ by affecting the drift and
in the first vertical derivative of the functional affecting the diffusion
coefficient.

\subsection{\label{subsec:Expected-future-effects}Expected future effects of
the emissions paths}

The most challenging part of the paper is to characterize the dynamics
of the expected future marginal effects of current emission represented
by the conditional expectation process $E\left[\int_{t}^{T}k_{t,s}(w^{s})ds\bigg|\mathcal{F}_{t}\right].$
Let $\xi_{t}=\int_{t}^{T}k_{t,s}(w^{s})ds$. Note that the conditional
expectation process $\mathbb{E}\bigg[\xi_{t}|\mathcal{F}_{t}\bigg]$
is in general not a martingale. For example, if $k_{s,t}=1,$ then
$\xi_{t}=T-t$ and we have $\mathbb{E}\bigg[\xi_{t}|\mathcal{F}_{t}\bigg]=T-t$.
Also note that $\xi_{t}$ is not adapted in general. The difficulty
more broadly comes from the fact that in the conditional expectation
both the process $\xi_{t}$ and the filtration $\mathcal{F}_{t}$
are changing with time. 

\subsubsection{\label{subsec:Total-derivative-formula}Total derivative formula
for conditional expectations}

In this section, we derive a new total derivative formula for conditional
expectation processes -- this is the main new theoretical tool that
we develop in this paper. A more general statement of the result and
a complete proof is given in the Appendix. Here, we state a simpler
version of the result and an outline of a heuristic proof.
\begin{defn}
The Malliavin derivative of the cylindrical functional $f(w)=g(w_{t_{1}},\ldots,w_{t_{k}}),$
where $t_{1}<\ldots<t_{k}$ and $g$ is a smooth function, is defined
by $D_{t}f(w)=\sum_{i=1}^{k}\partial_{i}g(w_{t_{1}},\ldots,w_{t_{k}})1_{t\leq t_{i}}.$
This defines a closable linear operator $D:L^{2}(\Omega)\to L^{2}(\Omega\times[0,T]).$
Its closure is also denoted by $D$ and is called the Malliavin derivative
operator. 
\end{defn}
Intuitively, a Malliavin derivative is a change in a functional due
to a perturbation of the whole path of the process (Figure \ref{fig:Malliavin derivative}).

\begin{figure}[h]
\centering{}\includegraphics[scale=3]{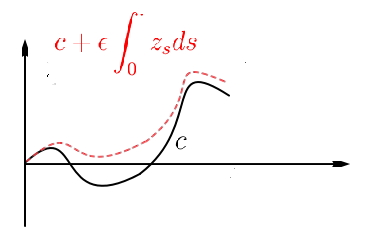}\caption{\label{fig:Malliavin derivative}Malliavin derivative: path perturbation}
\end{figure}

The next proposition states the total derivative formula, the main
technical tool that we develop in this paper.
\begin{prop}
\label{prop: Total derivative formula} \textbf{(Total derivative
formula) }Let an adapted square integrable process $(X_{t})_{0\leq t\leq T}$
be represented as
\[
X_{t}=\mathbb{E}[\xi_{t}|\mathcal{F}_{t}].
\]
Let $(\xi_{t})_{0\leq t\leq T}$ be a Malliavin differentiable square
integrable absolutely continuous process, that is, a process of the
form $\xi_{t}=\xi_{0}+\int_{0}^{t}\eta_{s}ds,$ where $\xi$ and $\eta$
may be anticipative. Then, 
\begin{equation}
dX_{t}=\mathbb{E}[\partial_{t}\xi_{t}|\mathcal{F}_{t}]dt+\mathbb{E}[D_{t}\xi_{t}|\mathcal{F}_{t}]dw_{t},\label{eq:total derivative formula}
\end{equation}
where $D$ is a Malliavin derivative.
\end{prop}
\begin{proof}
The Clark-Ocone formula gives
\[
\xi_{t}=\mathbb{E}\left[\xi_{t}\right]+\int_{0}^{T}\mathbb{E}[D_{r}\xi_{t}|\mathcal{F}_{r}]dw_{r}.
\]
Taking the conditional expectation leads to
\[
X_{t}=\mathbb{E}\left[\xi_{t}\right]+\int_{0}^{t}\mathbb{E}[D_{r}\xi_{t}|\mathcal{F}_{r}]dw_{r}.
\]
The total derivative is then
\begin{align*}
dX_{t} & =\mathbb{E}\left[\partial_{t}\xi_{t}\right]dt+\mathbb{E}[D_{t}\xi_{t}|\mathcal{F}_{t}]dw_{t}+\left(\int_{0}^{t}\mathbb{E}[D_{r}\partial_{t}\xi_{t}|\mathcal{F}_{r}]dw_{r}\right)dt=\\
 & =\mathbb{E}[D_{t}\xi_{t}|\mathcal{F}_{t}]dw_{t}+\left(\mathbb{E}\left[\partial_{t}\xi_{t}\right]+\int_{0}^{t}\mathbb{E}[D_{r}\partial_{t}\xi_{t}|\mathcal{F}_{r}]dw_{r}\right)dt.
\end{align*}
Note now that by the Clark-Ocone formula applied to the derivative
$\partial_{t}\xi_{t}$: 
\[
\partial_{t}\xi_{t}=\mathbb{E}\left[\partial_{t}\xi_{t}\right]+\int_{0}^{T}\mathbb{E}[D_{r}\partial_{t}\xi_{t}|\mathcal{F}_{r}]dw_{r},
\]
and taking the conditional expectations:
\[
\mathbb{E}[\partial_{t}\xi_{t}|\mathcal{F}_{t}]=\mathbb{E}\left[\partial_{t}\xi_{t}\right]+\int_{0}^{t}\mathbb{E}[D_{r}\partial_{t}\xi_{t}|\mathcal{F}_{r}]dw_{r}.
\]
The total derivative is then given by (\ref{eq:total derivative formula}).
\end{proof}
The lemma provides a time-dependent extension of the Clark-Ocone formula
for processes that can be represented as a conditional expectation
of an absolutely continuous process.\footnote{See Section \ref{subsec:Relationship-to-the_Clarrk-Ocone} for the
more detailed discussion of the relationship of our total derivative
formula to the Clark-Ocone formula.} The first term, $\mathbb{E}[\partial_{t}\xi_{t}|\mathcal{F}_{t}]$,
can be thought of as a time derivative of this variable and represents
how the conditional expectation evolves with respect to time. The
second term, $\mathbb{E}[D_{t}\xi_{t}|\mathcal{F}_{t}]$, can be thought
of as a stochastic derivative with respect to the underlying process
$w_{t}$ and represents how the conditional expectation changes with
the changes in $w_{t}$.\footnote{The process $(\xi_{t})$ in Proposition \ref{prop: Total derivative formula}
need not be adapted, and we can understand the process $X_{t}=\mathbb{E}[\xi_{t}|\mathcal{F}_{t}]$
as the adapted projection of the process $\xi$. In Proposition \ref{prop: Total derivative formula}
we thus prove that an adapted projection of an absolutely continuous
process is necessarily an Itô process. In the Appendix \ref{subsec:Characterization-of-It=0000F4}
we show the converse to this clam -- any Itô process can be represented
as such projection.}

\subsubsection{Semimartingale decomposition of the expected future effects of policy
paths}

We now make mild assumptions on the kernel $k_{t,s}$ -- absolute
continuity of the kernel in variable $t$ and Malliavin differentiability
as a functional from $w^{s}$ -- and apply the total derivative formula
to the expected future marginal damages, $E\left[\int_{t}^{T}k_{t,s}(w^{s})ds\bigg|\mathcal{F}_{t}\right]$:
\begin{align*}
dE\left[\int_{t}^{T}k_{t,s}(w^{s})ds\bigg|\mathcal{F}_{t}\right] & =\underset{\text{time derivative}}{\underbrace{\mathbb{E}\bigg[\partial_{t}\int_{t}^{T}k_{t,s}(w^{s})ds\bigg|\mathcal{F}_{t}\bigg]}}dt+\underset{\text{stochastic derivative}}{\underbrace{\mathbb{E}\bigg[D_{t}\int_{t}^{T}k_{t,s}(w^{s})ds\bigg|\mathcal{F}_{t}\bigg]}}dw_{t}=\\
 & =\underset{\text{time derivative}}{\underbrace{\mathbb{E}\bigg[\int_{t}^{T}\partial_{t}k_{t,s}(w^{s})ds\bigg|\mathcal{F}_{t}\bigg]-k_{t,t}\left(w^{t}\right)}}dt+\underset{\text{stochastic derivative}}{\underbrace{\mathbb{E}\bigg[\int_{t}^{T}D_{t}k_{t,s}(w^{s})ds\bigg|\mathcal{F}_{t}\bigg]}}dw_{t}.
\end{align*}
The meaning of this equation is straightforward. The conditional expectation
$dE\left[\int_{t}^{T}k_{t,s}(w^{s})ds\bigg|\mathcal{F}_{t}\right]$
measures the effects of the change in the emissions in period $t$
on the expected future marginal externality $k_{t,s}(w^{s})$ in all
periods $s$ ($s\in\left[t,T\right])$. The time evolution of the
conditional expectation is given by the analogue of the time derivative.
The stochastic evolution of the conditional expectation is given by
a stochastic derivative of the cumulative change in all marginal effects
of emission in time $t$ on all future periods $s$ with respect to
variation in the underlying process $w$.

\subsection{\label{subsec:Dynamics Pigouvian Climate}Dynamics of the marginal
externality damage}

We can now collect the results of the previous two sections. We obtain
the dynamics of the marginal externality damage 
\begin{align*}
d\Lambda_{t} & =-\epsilon\left(\Delta_{t}g_{t}(w^{t})+\frac{1}{2}\partial_{c_{t}}^{2}g_{t}(w^{t})+\mathbb{E}\bigg[\int_{t}^{T}\partial_{t}k_{t,s}(w^{s})ds\bigg|\mathcal{F}_{t}\bigg]-k_{t,t}\left(w^{t}\right)\right)dt+\\
 & +\left(\partial_{c_{t}}g_{t}(w^{t})+\mathbb{E}\bigg[\int_{t}^{T}D_{t}k_{t,s}(w^{s})ds\bigg|\mathcal{F}_{t}\bigg]\right)dw_{t}.
\end{align*}
This, we believe, is an important result as it solves in the closed
form the dynamics of the evolution of the marginal externality damage
or the optimal Pigouvian tax. This formula shows how the policymaker
should change the tax with the passage of time and with the realizations
of uncertainty -- that is, we determined the drift and the diffusion
coefficients of the process of the optimal Pigouvian tax. 

A related way to think about this result is that it determines what
features of the model matter up to the first order. For the time updating,
the functional Itô formula shows that the horizontal and the second
order vertical derivative are important and the total derivative formula
shows that the ``time derivative'' of the expected cumulative damages
are important. For the updating with the movement of uncertainty,
the functional Itô formula shows that the vertical derivative and
the total derivative formula shows that the ``stochastic derivative''
of the expected cumulative marginal damages are important.

We also immediately obtain the dynamics of the optimal emission policies
from equation (\ref{eq:FOC emissions}) in closed form:
\begin{align*}
dc_{t} & =-\epsilon\left(\Delta_{t}g_{t}(w^{t})+\frac{1}{2}\partial_{c_{t}}^{2}g_{t}(w^{t})+\mathbb{E}\bigg[\int_{t}^{T}\partial_{t}k_{t,s}(w^{s})ds\bigg|\mathcal{F}_{t}\bigg]-k_{t,t}\left(w^{t}\right)\right)dt+\\
 & +\left(1-\epsilon\partial_{c_{t}}g_{t}(w^{t})-\epsilon\mathbb{E}\bigg[\int_{t}^{T}D_{t}k_{t,s}(w^{s})ds\bigg|\mathcal{F}_{t}\bigg]\right)dw_{t}.
\end{align*}

We thus extended the result of the macro-climate model of Golosov,
et al. (2014) to the case with the path-dependent hysteresis given
by (\ref{eq: Climate hysteresis functional}). Note that all of the
results here are given in closed form. The next sections show that
the main insights derived with this specifications apply to general
hysteresis functionals.

\section{\label{sec:General-setting}General hysteresis functionals}

We now return to the analysis of the problem (\ref{eq:Perturbed problem})
for a general hysteresis functional.

\subsection{\label{subsec:Class A_t}The class $\mathcal{A}_{t}$ of path-dependent
functionals}

In this section we describe a class of path-dependent functionals
$\mathcal{A}_{t}$ which, in a sense, is parallel to continuously
differentiable functions of a real-valued argument. It will turn out
that this class, while very general, allows significant tractability
in the analysis of path-dependent problems. The introduction of this
class of functionals is one contribution of our paper.
\begin{rem}
\label{rem:Discontinuous}Adding a path-dependent functional $h_{t}$
can considerably change the structure of the optimal policy. For example,
consider the functional $h_{t}(c)=c_{\frac{t}{2}}$ that describes
the effects of the policies at the middle of the time period $\left[0,t\right]$.
We show in the appendix that the optimal policy is no longer continuous
and has a jump at $t=\frac{T}{2}$.
\end{rem}
Let us recall that a Fréchet derivative of a functional $h_{t}:C[0,t]\to\mathbb{R}$
at a point $c\in C[0,t]$ is a signed measure $\mu$ on $[0,t]$ such
that
\[
h_{t}\left(c+z\right)=h_{t}\left(c\right)+\int_{0}^{t}z_{s}d\mu\left(t,s\right)+o\left(\left\Vert z\right\Vert \right),\left\Vert z\right\Vert \rightarrow0.
\]
The primary difficulty with this formulation is that it features a
very general dependence of the measure $\mu$ on time $t$. The next
assumption structures this dependence.
\begin{assumption}
\label{assu:Class A}\textbf{(Class $\mathcal{A}_{t}$)} Suppose the
(Fréchet) derivative of the functional $h_{t}$ has an absolutely
continuous part and an atom at point $t$:
\[
h_{t}(c+z)=h_{t}(c)+\partial_{c_{t}}h_{t}(c)z_{t}+\int_{0}^{t}\delta_{s}h_{t}(c)z_{s}ds+o\left(\left\Vert z\right\Vert \right),\left\Vert z\right\Vert \rightarrow0.
\]
The family of such functionals is denoted by $\mathcal{A}_{t}.$
\end{assumption}
This assumption on the derivative of functionals means that there
are two sources of variation for $h_{t}$: an instantaneous influence
of a perturbation at the moment $t$ given by an atom and an integral
influence of perturbations at previous moments $s\le t$ which is
an absolutely continuous process. Assumption \ref{assu:Class A} is
mild -- it only requires that the derivative of a functional (which
is a measure) is absolutely continuous and has an atom at the present
time.

The class $\mathcal{A}_{t}$ contains a variety of functionals:
\[
h_{t}(c)=f(c_{t}),
\]
 which is state-dependent but not path-dependent;
\[
h_{t}(c)=f_{t}(c_{t},\int_{0}^{t}g_{t}(c_{s})ds),
\]
which jointly depends on state $c_{t}$ and the integral influence
of the path $\int_{0}^{t}g_{t}(c_{s})ds$ and, moreover, both the
joint dependence $f_{t}$ and the effects of the past policies $g_{t}$
depend on time $t$;
\[
h_{t}\left(c\right)=f_{t}(\int_{0}^{t}\int_{0}^{t}g_{t}(c_{s},c_{r})dsdr,...)
\]
where now there is a joint dependence on the past values $c_{s}$
and $c_{r}$ via a repeated integral.

The proposition that follows shows that this assumption covers a very
general class of functionals.
\begin{prop}
\label{claim:Class A is dense}Every Fréchet differentiable functional
$g:C[0,T]\to\mathbb{R}$ is a pointwise limit of functionals from
the class $\mathcal{A}_{T}.$
\end{prop}
\begin{proof}
In the appendix.
\end{proof}
We next explore the nature of the present marginal influence of policy
$\partial_{c_{t}}h_{t}$ which is itself a path-dependent functional.
The lemma that follows (proven for a more general case in the appendix)
connects Dupire's vertical derivative to the derivatives of the functionals
in the class $\mathcal{A}_{t}$.
\begin{lem}
\label{lem: Vertical derivative}Let the functional $h_{t}$ be in
the class $\mathcal{A}_{t}$ , then $\partial_{c_{t}}h_{t}(c)$ is
the vertical derivative in the Dupire sense.
\end{lem}
\begin{rem}
\label{rem:Integration by parts} We now discuss some additional motivation
behind Assumption \ref{assu:Class A}. Consider, for example, hysteresis
given by an Itô process:
\begin{equation}
h_{t}(c)=\int_{0}^{t}b_{s}dc_{s},\label{eq: h as Ito process}
\end{equation}
where $b$ is an absolutely continuous function. The functional $h_{t}(c)$
in equation (\ref{eq: h as Ito process}) is well-defined if and only
if the process $(c_{s})_{0\leq s\leq t}$ is of bounded variation.
This is not true for an arbitrary progressively measurable process
$(c_{s})_{0\leq s\leq t}$, and we cannot expect that the optimal
policy process will be of bounded variation. For example, in our baseline
case of Section \ref{subsec:Baseline-setting} $c_{t}^{*}=w_{t}$
and then almost all realizations of $c^{*}$ are of unbounded variation.
However, if the coefficient $b$ is absolutely continuous, then we
can integrate (\ref{eq: h as Ito process}) by parts and rewrite $h_{t}(c)$
in the form:
\[
h_{t}(c)=-\int_{0}^{t}b'_{s}c_{s}ds+b_{t}c_{t}.
\]
Further, $h_{t}$ is just a linear functional from $c:$
\[
h_{t}(c+z)-h_{t}(c)=b_{t}z_{t}-\int_{0}^{t}b'_{s}z_{s}ds.
\]
So, $h_{t}\in\mathcal{A}_{t}$ with $\partial_{c_{t}}h(c)=b_{t},$
$\delta_{s}h_{t}(c)=-b_{s}'.$
\end{rem}

\subsection{\label{sec:Optimal-policy}Optimal policy: the first order conditions}

In this section, we derive the first order conditions for the optimal
process $c^{\epsilon}$. While in some circumstances one can write
a recursive formulation even for a path-dependent problem and then
find a Hamilton-Jacobi-Bellman equation, in general it is difficult
or impossible to do it. Here, we instead find the first order conditions
for optimal policy using a variational method.
\begin{prop}
\textbf{(The first order conditions for the optimum)} Let the optimal
policy process $c^{\epsilon}$ solve (\ref{eq:Perturbed problem}),
then
\begin{equation}
c_{t}^{\epsilon}=w_{t}-\epsilon\left(\partial_{c_{t}}h_{t}(\left(c^{\epsilon}\right)^{t})+E\left[\int_{t}^{T}\delta_{t}h_{s}(\left(c^{\epsilon}\right)^{s})ds\bigg|\mathcal{F}_{t}\right]\right).\label{eq: FOC, perturbed}
\end{equation}
\end{prop}
\begin{proof}
Perturb the process $c$ by $\nu z,$ where $z$ is an adapted process,
and compute the derivative in $\nu$ at $\nu=0:$ 
\begin{align*}
 & \partial_{\nu}E\int_{0}^{T}\left(-\frac{1}{2}(c_{t}+\nu z_{t}-w_{t})^{2}-\epsilon h_{t}(c^{t}+\nu z^{t})\right)dt\bigg|_{\nu=0}=\\
 & =E\int_{0}^{T}\left(z_{t}w_{t}-c_{t}z_{t}-\epsilon\left(\partial_{c_{t}}h_{t}(c^{t})z_{t}+\int_{0}^{t}\delta_{s}h_{t}(c^{t})z_{s}ds\right)\right)dt=\\
 & =E\int_{0}^{T}z_{t}\left(w_{t}-c_{t}-\epsilon\left(\partial_{c_{t}}h_{t}(c^{t})+\int_{t}^{T}\delta_{t}h_{s}(c^{s})ds\right)\right)dt=0.
\end{align*}
Since $z$ is an arbitrary adapted process, we get the first order
conditions (\ref{eq: FOC, perturbed}). The importance of Assumption
\ref{assu:Class A} is evident in particular in the third line of
the proof. If we did not impose this assumption, then instead of the
integral $\int_{t}^{T}\delta_{t}h_{s}(c^{s})ds$ we would have a general
measure $\mu_{t}(c^{s},[t,T])$ which can even be discontinuous in
$t$. Assumption \ref{assu:Class A} imposes a smooth structure for
this measure by requiring that the measure has an absolutely continuous
part.
\end{proof}
This equation (\ref{eq: FOC, perturbed}) has a natural economic meaning.
When hysteresis $\epsilon h_{t}$ is present, the solution to the
optimal problem ($c_{t}^{*}=w_{t})$ is modified by the two terms.
The first term, 
\[
I_{t}=\epsilon\partial_{c_{t}}h_{t}(c^{\epsilon}),
\]
is the instantaneous marginal effect of the change in the policy $c_{t}^{\epsilon}$
on functional $h_{t}$. The second term gives the future marginal
effects of policy $c_{t}^{\epsilon}$
\[
F_{t}=\epsilon E\left[\int_{t}^{T}\delta_{t}h_{s}((c^{\epsilon})^{s})ds\bigg|\mathcal{F}_{t}\right].
\]
 on the values of all future hysteresis functionals $h_{s}\left(\left(c^{\epsilon}\right)^{s}\right)$.
For each time $s$ (where $t\le s\le T$), the marginal effect of
changing policy in period $t$ is represented by the derivative $\delta_{t}h_{s}(\left(c^{\epsilon}\right)^{s})$
of the functional $h_{s}$ with respect to change in policy at time
$t$. These marginal effects are evaluated as a conditional expectation
at time $t$ and hence represent the expected future marginal effects.
It is useful to think about the term $F_{t}$ as the cumulative impulse
response of the change of the policy today on all future hysteresis
functionals. This is similar to Alvarez and Lippi (2019), Alvarez,
Le Bihan, and Lippi (2016), Alvarez, Lippi, and Oskolkov (2020) and
Borovi\v{c}ka, Hansen, Scheinkman (2014), Borovi\v{c}ka, Hansen, Hendricks,
and Scheinkman (2011), and Borovi\v{c}ka and Hansen (2016). Note that
$F_{t}$ is a conditional expectation process that may non-trivially
change with time as both the marginal effects $\delta_{t}h_{s}$ and
the filtration $\mathcal{F}_{t}$ changes.

Both terms depend on the path of policy $c_{\left[0,s\right]}^{\epsilon}=(c^{\epsilon})^{s}$.
As seen here, the assumption that the functional $h_{t}$ is in the
class $\mathcal{A}_{t}$ allows us to conveniently separate the first
order condition in two parts: the marginal effects of the policy on
current period $I_{t}$ and the conditional expectation of the cumulative
future marginal effects of policy paths in the future, $F_{t}$. 

\section{Characterizing the general problem}

This section is divided into two main parts. Section \ref{sec:First-order-process}
derives a closed-form characterization of the change in the optimal
policy when hysteresis is small. Section \ref{subsec:Dynamics-of-optimal}
provides a characterization of the optimal policy.\footnote{It may be useful for a reader to also consider an example in the appendix
that provides a parallel characterization for the case with no path-dependency.} 

\subsection{\label{sec:First-order-process}Stochastic hysteresis elasticity
and its dynamics}

In this section we derive a closed-form characterization of the change
in the optimal policy when hysteresis is small. We call this first
order process a stochastic hysteresis elasticity. We characterize
the dynamics of the stochastic elasticity in closed form by providing
its semimartingale decomposition which is the main contribution of
this section.

\subsubsection{\label{subsec:FOP}Stochastic hysteresis elasticity}

We first formally define the stochastic hysteresis elasticity. This
is the first order process that represents the change in optimal policy
process $c^{*}$ in response to introduction of an infinitesimal path-dependent
hysteresis functional $h_{t}$ . 
\begin{defn}
\textbf{(Stochastic hysteresis elasticity)} Let $c_{t}^{*}$ be a
solution to the baseline problem (\ref{eq:Unperturbed problem}) and
$c_{t}^{\epsilon}$ be a solution to the problem (\ref{eq:Perturbed problem}),
where hysteresis is given by the functional $\epsilon h_{t}$, $\epsilon\rightarrow0$.
The first-order process or the stochastic hysteresis elasticity $C_{t}^{h}$
is such that
\[
c_{t}^{\epsilon}=c_{t}^{*}+\epsilon C_{t}^{h}+o(\epsilon).
\]
\end{defn}
The stochastic elasticity $C_{t}^{h}$ is the change to the first
order in the optimal policy process in response to a small change
in the hysteresis $h$. In this sense, it is similar to the usual
concept of elasticity but now determines how the whole process $c_{t}$
changes. %
{} From now on, to ease notation we drop the dependence of $C_{t}^{h}$
on $h$ and denote it simply by $C_{t}$.

Differentiating (\ref{eq: FOC, perturbed}) with respect to $\epsilon$
at $\epsilon=0$, and recalling that $c_{t}^{0}=c_{t}^{*}=w_{t}$
we find the process $C_{t}$ explicitly
\begin{equation}
C_{t}=-\partial_{c_{t}}h_{t}(w^{t})-E\left[\int_{t}^{T}\delta_{t}h_{s}(w^{s})ds\bigg|\mathcal{F}_{t}\right].\label{eq: First order process}
\end{equation}
This equation is already interesting by itself as it presents the
change in the optimal policy plan due to the introduction of the path-dependent
effect. Moreover, it is given in closed form.

We now turn to presenting the main result of this section -- showing
that $C_{t}$ is an Itô process and, most importantly, writing its
explicit semimartingale decomposition. We are seeking a representation
\[
dC_{t}=\beta\left(w^{t}\right)dt+\gamma\left(w^{t}\right)dw_{t},
\]
where $\beta$ and $\gamma$ are potentially path-dependent. In other
words, we want to determine the dynamics of the stochastic elasticity
$C_{t}$.

\subsubsection{\label{subsec:Present-effect-of}Present effect of the policy path }

This section provides a semimartingale decomposition of the term $\partial_{c_{t}}h_{t}(c^{*})$
-- the present effect of the path of policies $(c^{*})^{t}=w^{t}$.
The principal tool that we use is Dupire's functional Itô formula
in Lemma \ref{lem:Functional-It=0000F4's-Formula:} that applies to
the non-anticipative functionals of the path. We now make an assumption
that allows the use of this formula.
\begin{assumption}
There exists a non-anticipative functional $q:[0,T]\times D[0,T]\to\mathbb{R}$
which is horizontally differentiable, twice vertically differentiable
and is such that $q_{t}(c)=\partial_{c_{t}}h_{t}(c).$ In other words,
the derivative \textup{$\partial_{c_{t}}h_{t}(c)$} can be extended
to a $C_{b}^{1,2}$-functional on the space of cádlág paths. We denote
this extension also by $\partial_{c_{t}}h_{t}(c)$.
\end{assumption}
Applying the functional Itô formula, and noting that $c_{t}^{*}=w_{t}$
we obtain a semimartingale decomposition of the functional $\partial_{c_{t}}h_{t}(c^{*})$
that depends on the whole history of policy $(c^{*})^{t}=w^{t}$:
\[
d(\partial_{c_{t}}h_{t}(c^{*}))=\Delta_{t}\left(\partial_{c_{t}}h_{t}(w^{t})\right)dt+\partial_{c_{t}}\left(\partial_{c_{t}}h_{t}(w^{t})\right)dw_{t}+\frac{1}{2}\partial_{c_{t}}\left(\partial_{c_{t}}^{2}h_{t}(w^{t})\right)dt,
\]
and gathering the terms we obtain the following lemma.
\begin{lem}
\label{claim:Dynamics Present}\textbf{(Dynamics of the present effect,
$I_{t}$)} The semimartingale decomposition of the present marginal
effects, $I_{t}$, of the policy path is given by
\begin{align}
d(\partial_{c_{t}}h_{t}(c^{*})) & =\left(\Delta_{t}(\partial_{c_{t}}h_{t}(w^{t}))+\frac{1}{2}\partial_{c_{t}}^{3}h_{t}(w^{t})\right)dt+\partial_{c_{t}}^{2}h_{t}(w^{t})dw_{t}.\label{eq: FOP, Dupire part}
\end{align}
\end{lem}
This part of the derivation is already interesting as a stand-alone
result. The reason why one can apply the Dupire and Cont-Fournie analysis
is that $\partial_{c_{t}}h_{t}(c^{*})$ is already represented as
a functional of the path. The decomposition (\ref{eq: FOP, Dupire part})
then has the same intuitive meaning as the standard Itô's formula
but now applies to the functional of the past, not the function of
the present realization. Note that we are already applying the functional
Itô's formula to the marginal effects, that is, to the vertical derivative
of the functional $\partial_{c_{t}}h_{t}(w^{t})$. Hence, there are
the second, $\partial_{c_{t}}^{2}$, and the third, $\partial_{c_{t}}^{3}$,
vertical derivatives as well as the mixed derivative $\Delta_{t}\partial_{c_{t}}$.

\subsubsection{\label{subsec:Semimartingale-Future}Expected future effects of policy
paths}

We now make an assumption on the derivative of the functional, $\delta_{t}h_{s}(c^{*})$,
that is needed to apply the total derivative formula of Proposition
\ref{prop: Total derivative formula}.\footnote{In Section \ref{subsec:Assumptions-on-smoothness}, we further discuss
the smoothness assumptions for the functionals $h_{t}$.}
\begin{assumption}
The derivative $\delta_{t}h_{s}(c^{*})$ is an absolutely continuous
in $t$ and is a Malliavin differentiable functional of the path $w^{s}.$
\end{assumption}
It follows that the future marginal effect of a given policy path
is given by
\[
\frac{d}{dt}\int_{t}^{T}\delta_{t}h_{s}(c^{*})ds=-\delta_{t}h_{t}(c^{*})+\int_{t}^{T}\frac{\partial}{\partial t}(\delta_{t}h_{s})(c^{*})ds.
\]
and its conditional expectation is given by the total derivative formula
(\ref{eq:total derivative formula}):
\begin{align*}
d\mathbb{E}\bigg[\int_{t}^{T}\delta_{t}h_{s}(c^{*})ds\bigg|\mathcal{F}_{t}\bigg] & =\underset{\text{time derivative}}{\underbrace{\mathbb{E}\bigg[\partial_{t}\int_{t}^{T}\delta_{t}h_{s}(c^{*})ds\bigg|\mathcal{F}_{t}\bigg]}}dt+\underset{\text{stochastic derivative}}{\underbrace{\mathbb{E}\bigg[D_{t}\int_{t}^{T}\delta_{t}h_{s}(c^{*})ds\bigg|\mathcal{F}_{t}\bigg]}}dw_{t}.
\end{align*}
The meaning of this equation is straightforward. The conditional expectation
\[
F_{t}=\mathbb{E}\bigg[\int_{t}^{T}\delta_{t}h_{s}(c^{*})ds\bigg|\mathcal{F}_{t}\bigg]
\]
describes how a change in policy in period $t$ determines the expected
future marginal effects of that policy $\delta_{t}h_{s}(c^{*})$ in
all periods $s$ ($s\in\left[t,T\right])$. The time evolution of
the conditional expectation is given by the analogue of the time derivative.
The stochastic evolution of the conditional expectation is given by
a stochastic derivative of the cumulative change in all marginal effects
of policy $t$ on all future periods $s$ with respect to variation
in the underlying process $w$. Collecting the terms gives us the
dynamics of the future effects of policy and shows that it is an Itô
process.
\begin{lem}
\label{lem:Dynamics Future}\textbf{(Dynamics of the future effects,
$F_{t}$)} The semimartingale representation of the future marginal
effects, $F_{t}$, of policy paths is given by:
\begin{align}
 & d\mathbb{E}\bigg[\int_{t}^{T}\delta_{t}h_{s}(c^{*})ds\bigg|\mathcal{F}_{t}\bigg]=\nonumber \\
 & =\Bigl(-\delta_{t}h_{t}(w^{t})+\int_{t}^{T}\mathbb{E}\left[\frac{\partial}{\partial t}(\delta_{t}h_{s})(w^{s})\bigg|\mathcal{F}_{t}\right]ds\Bigr)dt+\mathbb{E}\bigg[D_{t}\int_{t}^{T}\delta_{t}h_{s}(w^{s})ds\bigg|\mathcal{F}_{t}\bigg]dw_{t}.\label{eq: FOP, total derivative part}
\end{align}
\end{lem}

\subsubsection{Dynamics of the stochastic elasticity}

We now combine the results of the decompositions of the present effects
(\ref{eq: FOP, Dupire part}) and the future effects (\ref{eq: FOP, total derivative part})
of policy and recalling that $c_{t}^{*}=w_{t}$:
\begin{align*}
dC_{t}= & -\Bigl(\underset{\text{Present effects: functional Itô}}{\underbrace{\Delta_{t}\partial_{c_{t}}h_{t}(w^{t})dt+\frac{1}{2}\partial_{c_{t}}^{3}h_{t}(w^{t})dt+\partial_{c_{t}}^{2}h_{t}(w^{t})dw_{t}}}\Bigr)-\\
 & -\Bigl(\underset{\text{Future effects: total derivative}}{\underbrace{-\delta_{t}h_{t}(w^{t})dt+E\left[\int_{t}^{T}\frac{\partial}{\partial t}(\delta_{t}h_{s})(w^{s})ds\bigg|\mathcal{F}_{t}\right]dt+E\left[\int_{t}^{T}D_{t}\left(\delta_{t}h_{s}(w^{s})\right)ds\bigg|\mathcal{F}_{t}\right]dw_{t}}}\Bigr),
\end{align*}
where in the last term we interchanged the Malliavin derivative and
the integral.

Collecting the terms gives the semimartingale decomposition of the
first order process $C_{t}$ in closed form.
\begin{thm}
\label{thm:First order process, main theorem}\textbf{(Dynamics of
stochastic elasticity)} The semimartingale decomposition of the stochastic
elasticity $C_{t}$ is given by 
\begin{align}
dC_{t}= & -\Bigl(\Delta_{t}\partial_{c_{t}}h_{t}(w^{t})+\frac{1}{2}\partial_{c_{t}}^{3}h_{t}(w^{t})-\delta_{t}h_{t}(w^{t})+E\left[\int_{t}^{T}\frac{\partial}{\partial t}(\delta_{t}h_{s})(w^{s})ds\bigg|\mathcal{F}_{t}\right]\Bigr)dt-\nonumber \\
 & -\Bigl(\partial_{c_{t}}^{2}h_{t}(w^{t})+E\left[\int_{t}^{T}D_{t}\left(\delta_{t}h_{s}(w^{s})\right)ds\bigg|\mathcal{F}_{t}\right]\Bigr)dw_{t}.\label{eq: FOP, main formula}
\end{align}
\end{thm}
This theorem shows that the optimal process $c_{t}^{*}$ changes with
the introduction of hysteresis, that is, the additional infinitesimal
path-dependent effect of policies. The equation (\ref{eq: FOP, main formula})
gives the first order process in closed form as
\[
dC_{t}^{h}=\beta_{t}\left(w^{t}\right)dt+\gamma_{t}\left(w^{t}\right)dw_{t},
\]
where $\beta$ and $\gamma$ are path-dependent coefficients.

\subsection{\label{subsec:Dynamics-of-optimal}Dynamics of optimal policy}

We now characterize the dynamics of the optimal policy. 

In this section, we assume that $\epsilon=1$, without loss of generality.
The first order conditions (\ref{eq: FOC, perturbed}) become
\[
w_{t}=c_{t}+\partial_{c_{t}}h_{t}(c^{t})+E\left[\int_{t}^{T}\delta_{t}h_{s}(c^{s})ds\bigg|\mathcal{F}_{t}\right].
\]
Let us take the differential of this equation assuming that $c$ is
an Itô process, i.e. 
\[
dc_{t}=\alpha_{t}dt+\beta_{t}dw_{t}.
\]
The present marginal effects of the policy are differentiated using
the functional Itô formula (\ref{eq:Functional Ito}):
\[
d\left(\partial_{c_{t}}h_{t}(c^{t})\right)=\Delta_{t}\partial_{c_{t}}h_{t}(c^{t})dt+\partial_{c_{t}}^{2}h_{t}(c^{t})dc_{t}+\frac{1}{2}\partial_{c_{t}}^{3}h_{t}(c^{t})\beta_{t}^{2}dt.
\]
The future expected marginal effects of the policy are differentiated
using the total derivative formula (\ref{eq:total derivative formula}):
\begin{align*}
 & d\left(E\left[\int_{t}^{T}\delta_{t}h_{s}(c^{s})ds\bigg|\mathcal{F}_{t}\right]\right)=E\left[\partial_{t}\int_{t}^{T}\delta_{t}h_{s}(c^{s})ds\bigg|\mathcal{F}_{t}\right]dt+E\left[\int_{t}^{T}D_{t}\left(\delta_{t}h_{s}(c^{s})\right)ds\bigg|\mathcal{F}_{t}\right]dw_{t}=\\
 & =E\left[\int_{t}^{T}\partial_{t}\left(\delta_{t}h_{s}(c^{s})\right)ds\bigg|\mathcal{F}_{t}\right]dt-\delta_{t}h_{t}(c^{t})dt+E\left[\int_{t}^{T}D_{t}\left(\delta_{t}h_{s}(c^{s})\right)ds\bigg|\mathcal{F}_{t}\right]dw_{t}.
\end{align*}
The differential of the first order conditions is then given by
\begin{align*}
 & dw_{t}=dc_{t}+\Delta_{t}\partial_{c_{t}}h_{t}(c^{t})dt+\partial_{c_{t}}^{2}h_{t}(c^{t})dc_{t}+\frac{1}{2}\partial_{c_{t}}^{3}h_{t}(c^{t})\beta_{t}^{2}dt+\\
 & +E\left[\int_{t}^{T}\partial_{t}\delta_{t}h_{s}(c^{s})ds\bigg|\mathcal{F}_{t}\right]dt-\delta_{t}h_{t}(c^{t})dt+E\left[\int_{t}^{T}D_{t}\left(\delta_{t}h_{s}(c^{s})\right)ds\bigg|\mathcal{F}_{t}\right]dw_{t}.
\end{align*}
Collecting the terms near $dt$ and $dw_{t}$ we derive dynamics of
the optimal policy:
\[
\begin{cases}
0=(1+\partial_{c_{t}}^{2}h_{t}(c^{t}))\alpha_{t}+\frac{1}{2}\partial_{c_{t}}^{3}h_{t}(c^{t})\beta_{t}^{2}+\Delta_{t}\partial_{c_{t}}h_{t}(c^{t})+E\left[\int_{t}^{T}\partial_{t}\delta_{t}h_{s}(c^{s})ds\bigg|\mathcal{F}_{t}\right]-\delta_{t}h_{t}(c^{t}),\\
1=(1+\partial_{c_{t}}^{2}h_{t}(c^{t}))\beta_{t}+E\left[\int_{t}^{T}D_{t}\left(\delta_{t}h_{s}(c^{s})\right)ds\bigg|\mathcal{F}_{t}\right],
\end{cases}
\]
and the next theorem characterizes the drift and diffusion coefficients
of the optimal policy.
\begin{thm}
\label{prop:Optimal-policy}(\textbf{Optimal policy}) The dynamics
of optimal policy is given by:
\begin{align*}
\alpha_{t} & =-\frac{1}{1+\partial_{c_{t}}^{2}h_{t}(c^{t})}\left(\frac{1}{2}\partial_{c_{t}}^{3}h_{t}(c^{t})\beta_{t}^{2}+\Delta_{t}\partial_{c_{t}}h_{t}(c^{t})+E\left[\int_{t}^{T}\partial_{t}\delta_{t}h_{s}(c^{s})ds\bigg|\mathcal{F}_{t}\right]-\delta_{t}h_{t}(c^{t})\right),\\
\beta_{t} & =\frac{1-E\left[\int_{t}^{T}D_{t}\delta_{t}h_{s}(c^{s})ds\bigg|\mathcal{F}_{t}\right]}{(1+\partial_{c_{t}}^{2}h_{t}(c^{t}))}.
\end{align*}
\end{thm}
We can compare this theorem to the results without hysteresis, where
there are only contemporaneous effects of the policy $f\left(c_{t}\right)$,
the details of which are in the appendix. The diffusion coefficient
$\beta_{t}$ is determined by two terms. The first term, in the denominator,
contains the vertical derivative of the present marginal effects of
the policy $\partial_{c_{t}}\left(\partial_{c_{t}}h_{t}(c^{t})\right)$
and measures the change in the present marginal costs due to a ``bump''
in the trajectory of the costs. This is an analogue of the term $f''\left(c_{t}\right)$
for the case without hysteresis. The second term, in the numerator,
is the stochastic derivative $E\left[\int_{t}^{T}D_{t}\delta_{t}h_{s}(c^{s})ds\bigg|\mathcal{F}_{t}\right]$
which measures impact of the change in the trajectory of the stochastic
process on the expected future marginal effects in all periods. This
term is not present in the case without hysteresis.

The drift coefficient $\alpha_{t}$ is determined by three terms.
The first term, in the denominator, is the same relative scaling as
in the case of $\beta_{t}$. The second set of terms is the time evolution
of the present, $\Delta_{t}\partial_{c_{t}}h_{t}(c^{t})$, and future,
$E\left[\partial_{t}\int_{t}^{T}\delta_{t}h_{s}(c^{s})ds\bigg|\mathcal{F}_{t}\right]$,
marginal costs. The analogue of the horizontal derivative $\Delta_{t}\partial_{c_{t}}h_{t}(c^{t})$
for the case without hysteresis would be the time derivative $\partial_{t}f\left(c_{t},t\right)$.
The time evolution of the future marginal costs are is in the case
without hysteresis. The third term is the quadratic variance of the
present marginal costs, $\frac{1}{2}\partial_{c_{t}}^{3}h_{t}(c^{t})\beta_{t}^{2}$.
It is a path-dependent analogue of $\frac{1}{2}f'''\left(c_{t}\right)\left(\beta_{t}\right)^{2}$
in the case without hysteresis.

The equations Theorem \ref{prop:Optimal-policy} characterize the
drift and diffusion coefficients implicitly in contrast to the explicit
form we obtained in Theorem \ref{thm:First order process, main theorem}.
The explicit form we obtained for the stochastic elasticity thus can
be thought of as a characterization of the optimal process for small
hysteresis, or small hysteresis asymptotics.
\begin{rem*}
We now show that assuming that the marginal effect of policy $\delta_{t}h_{s}$
are in the class $\mathcal{A}_{s}$, we can compute Malliavin derivatives
$D_{t}\delta_{t}h_{s}(c^{s})ds$ in Theorem \ref{prop:Optimal-policy}
in a more detailed way.
\end{rem*}
\begin{cor*}
Let $\delta_{t}h_{s}$ be in the class $\mathcal{A}_{s}$, then 
\[
D_{t}\delta_{t}g_{s}(c^{s})=\partial_{c_{t}}\delta_{t}h_{s}(c^{s})D_{t}c_{s}+\int_{t}^{s}\delta_{r}\delta_{t}h_{s}(c^{s})D_{t}c_{r}dr,
\]
where the tangent process $D_{t}c_{r}$ is given by 
\begin{equation}
\begin{cases}
d(D_{t}c_{s})=(D_{t}\alpha_{s})ds+(D_{t}\beta_{s})dw_{s},\ s>t\\
D_{t}c_{t}=\beta_{t}.
\end{cases}\label{eq: Tangent process}
\end{equation}
\end{cor*}
\begin{proof}
Let us perturb the underlying Brownian motion: $S(w,\epsilon z)_{t}=w_{t}+\epsilon\int_{0}^{t}z_{r}dr,$
$0\leq t\leq T.$ Then by the definition of the Malliavin derivative
\[
c_{t}(S(w,\epsilon z))=c_{t}(w)+\epsilon\int_{0}^{t}D_{r}c_{t}(w)z_{r}dr+o(\epsilon),
\]
where $D_{r}c_{t}$ is a tangent process. Hence, 
\begin{align*}
 & \delta_{t}h_{s}(c^{s}(S(w,\epsilon z)))=\delta_{t}h_{s}(c^{s}(w))+\epsilon\partial_{c_{s}}\delta_{t}h_{s}(c^{s})\int_{0}^{s}D_{r}c_{s}(w)z_{r}dr+\\
 & +\epsilon\int_{0}^{s}\delta_{r}\delta_{t}h_{s}(c^{s})\int_{0}^{r}D_{u}c_{r}(w)z_{u}dudr+o(\epsilon)=\\
 & =\delta_{t}h_{s}(c^{s}(w))+\epsilon\int_{0}^{s}z_{r}\left(\partial_{c_{s}}\delta_{t}h_{s}(c^{s})D_{r}c_{s}(w)+\int_{r}^{s}\delta_{u}\delta_{t}h_{s}(c^{s})D_{r}c_{u}(w)du\right)dr+o(\epsilon).
\end{align*}
 The result follows.
\end{proof}
This characterization is useful as it uses the tangent process which
has a particularly simple form. When $c$ is an Itô process, its tangent
process is also an Itô process with the drift and diffusion coefficients
being the Malliavin derivatives of the drift and diffusion coefficients
of the process $c$.

\section{\label{sec:Examples-of-perturbation}Examples}

In this section, we present a number of examples of the hysteresis
functionals $h_{t}$ to illustrate how to use our theoretical results.
We also provide, whenever possible, an alternative derivation using
other tools. Section \ref{subsec:The-user's-guide} is the user's
guide that describes the steps needed to apply the theory we developed.
The first example, Section \ref{subsec:Example No-hysteresis:-state-dependent},
revisits the case of no hysteresis. In this case, only the present
effects of the policy are present and the functional Itô formula reduces
to the usual Itô formula. There is no need to use the total derivative
formula. The second example, Section \ref{subsec:Example Cumulative-hysteresis},
considers cumulative hysteresis. The effects of the polices are a
function of the current policy and the integral of the past policies.
For the case of the multiplicative dependence of the effects of the
current policies and the cumulative hysteresis, the stochastic elasticity
takes a very simple form. The present marginal effects of the policy
path only have the horizontal derivative which measures how the cumulant
of the past policies changes with time and all of the vertical derivatives
are equal to zero. The total derivative formula for the conditional
expectation of the future marginal effects only has the stochastic
derivative component which itself has a very simple form that measures
how lengthy the effects of the stochastic shock are. We expand on
this class of examples in Boulatov, Riabov, and Tsyvinski (2020) where
we analyze such cases extracting a martingale and not using the main
tools of this paper -- the functional Itô formula and the total derivative
formula. Our third example, Section \ref{subsec:Example Martingale},
studies hysteresis that depends both on time and the past policies.
This example captures three important parts of path-dependent policies.
First, there is an effect of the present policy. Second, there is
an effect of the hysteresis. The hysteresis is represented by an integral
with a kernel that has joint dependence on current time and past policies.
That is, path dependency sophisticatedly changes with both the time
and evolution of the past polices. Third, there is a joint dependence
between the present and the past. For this example, we need to utilize
all of the tools developed in the paper. The functional Itô formula
gives the evolution of the present marginal effects of the policies
in terms of horizontal and vertical derivatives. The total derivative
formula straightforwardly gives the evolution of the conditional expectation
process in terms of the time and stochastic derivative. We also consider
a deterministic setting with hysteresis in Section \ref{subsec: Example Deterministic-hysteresis}.
As there is no stochasticity, there is only evolution with respect
to time or, rather, the history of policies. We decompose its effects
into the effects of the present and the past history. The present
marginal effect of policy then evolves as the first-order horizontal
and the vertical derivatives. The future marginal effects of policies
evolve as their time derivative. In Section \ref{subsec:A-tipping-point}
we consider an example with a tipping point in which the time when
the stochastic process reaches the maximum becomes a reference point.

\subsection{\label{subsec:The-user's-guide}The user's guide}

In order to apply results from the previous section one needs to calculate
a number of derivatives of the functional $(h_{t})_{0\leq t\leq T}.$
In calculation of the first order conditions (\ref{eq: FOC, perturbed}),
we need to find a Fréchet derivative of $h_{t}$:
\[
h_{t}(c+z)-h_{t}(c)=\partial_{c_{t}}h_{t}(c)z_{t}+\int_{0}^{t}\delta_{s}h_{t}(c)z_{s}ds+o(||z||).
\]
For the application of Dupire's functional Itô's formula and finding
the SDE for the present effects of the path in Section \ref{subsec:Present-effect-of}
and equation (\ref{eq: FOP, Dupire part}) we need to calculate the
horizontal derivative $\Delta_{t}\partial_{c_{t}}h_{t}(c)$ and two
vertical derivatives $\partial_{c_{t}}^{2}h_{t}(c),$ $\partial_{c_{t}}^{3}h_{t}(c).$
For the application of the total derivative formula (\ref{eq: FOP, total derivative part})
and finding the evolution equation for the future effects of the path
we need to calculate two types of derivatives of $\delta_{t}h_{s}(c)$:
(a) time derivative $\frac{\partial}{\partial t}\delta_{t}h_{s}(c)$,
and (b) the Malliavin derivatives $D_{t}\delta_{t}h_{s}(c),$

In other words, we first break the first order condition for the effects
of policies of the present and the effects of the past. The dependence
of variables in the past implies that any change in a variable affects
all the future states. We then use the functional Itô formula to derive
the evolution equation for the effect on the present and the total
derivative formula for the effects on the future. The total derivative
formula also requires calculations of the Malliavin derivatives.

\subsection{\label{subsec:Example No-hysteresis:-state-dependent}No hysteresis:
state-dependent perturbations}

We start with the simplest case -- the state-dependent perturbation
where $h$ is a function of the current state rather than a functional.
Consider the perturbation $h_{t}(c)=f(c_{t}),$ where $f:\mathbb{R}\to\mathbb{R}$
is a smooth function. The list of derivatives becomes the following.

The Fréchet derivative of the functional is just the derivative of
$f:$ 
\[
h_{t}(c+z)-h_{t}(c)=f(c_{t}+z_{t})-f(c_{t})=f'(c_{t})z_{t}+o(||z||).
\]
Correspondingly, the present marginal effects of policy are given
by the regular derivative
\[
\partial_{c_{t}}h_{t}(c)=f'(c_{t}),
\]
and the effects of the past are 
\[
\delta_{s}h_{t}=0\mbox{ for }s<t.
\]
In particular, we see that $h_{t}\in\mathcal{A}_{t}$ and $\delta_{s}h_{t}\in\mathcal{A}_{t}.$

We now turn to the analysis of the present effect of the policy path.
The horizontal derivative $\Delta_{t}\left(\partial_{c_{t}}h_{t}(c)\right)$
is equal to zero. Indeed, it is defined as a limit 
\[
\Delta_{t}\partial_{c_{t}}h_{t}(c)=\lim_{\epsilon\to0}\frac{\partial_{c_{t+\epsilon}}h_{t+\epsilon}(c_{\cdot,\epsilon})-\partial_{c_{t}}h_{t}(c)}{\epsilon},
\]
where $c_{\cdot,\epsilon}$ is an extension of the path $c$ from
$[0,t]$ to $[0,t+\epsilon]$ by $c_{s,\epsilon}=c_{t}$ for $t\leq s\leq t+\epsilon.$
It follows that\footnote{If we had $h_{t}\left(c\right)=f\left(c_{t},t\right)$ the only difference
would be that $\Delta_{t}\partial_{c_{t}}h_{t}(c)=\partial_{c_{t},t}^{2}f\left(c_{t},t\right)$ }
\[
\Delta_{t}\partial_{c_{t}}h_{t}(c)=\lim_{\epsilon\to0}\frac{f'(c_{t+\epsilon,\epsilon})-f'(c_{t})}{\epsilon}=\lim_{\epsilon\to0}\frac{f'(c_{t})-f'(c_{t})}{\epsilon}=0.
\]
Vertical differentiation of the present marginal effect of policy
is the regular differentiation of $f:$
\[
\partial_{c_{t}}^{2}h_{t}(c)=f''(c_{t}),\ \partial_{c_{t}}^{3}h_{t}(c)=f'''(c_{t}).
\]

We now turn to the analysis of the expected future marginal effects.
The derivative in $t$ of the future marginal effects in period $s:$
$\delta_{t}h_{s}(c),$ i.e. $\frac{\partial}{\partial t}\delta_{t}h_{s}(c)$
is equal to zero, since $\delta_{t}h_{s}(c)=0:$
\[
\frac{\partial}{\partial t}\delta_{t}h_{s}(c)=0.
\]
The Malliavin derivative of $\delta_{t}h_{s}(c)$ is also zero, since
$\delta_{t}h_{s}(c)=0:$
\[
D_{t}\delta_{t}h_{s}(c)=0.
\]

The formula for the dynamics of the stochastic elasticity (\ref{eq: FOP, main formula})
becomes:
\begin{align*}
dC_{t}= & -\Bigl(\underset{=0}{\underbrace{\Delta_{t}\partial_{c_{t}}h_{t}(w^{t})}dt}+\frac{1}{2}\underset{=f'''\left(w_{t}\right)}{\underbrace{\partial_{c_{t}}^{3}h_{t}(w^{t})}}dt+\underset{=f''\left(w_{t}\right)}{\underbrace{\partial_{c_{t}}^{2}h_{t}(w^{t})}}dw_{t}\Bigr)-\\
 & -\Bigl(\underset{=0}{\underbrace{-\delta_{t}h_{t}(w^{t})dt+E\left[\int_{t}^{T}\partial_{t}\delta_{t}h_{s}(w^{s})ds\bigg|\mathcal{F}_{t}\right]dt+E\left[\int_{t}^{T}D_{t}\delta_{t}h_{s}(w^{s})ds\bigg|\mathcal{F}_{t}\right]dw_{t}}}\Bigr),
\end{align*}
and
\[
dC_{t}=-\frac{1}{2}f'''(w_{t})dt-f''(w_{t})dw_{t}.
\]

We summarize the result. First, the present effect of the path from
the Section \ref{subsec:Present-effect-of} are given by the standard
space derivatives. The functional Itô formula reduces to the usual
Itô's formula. Second, all of the future effects of the path from
are zero, hence, the conditional expectation is zero.

One can, of course, immediately get the result of this section by
applying the Itô formula to $f'(c_{t}^{*}),$ see the Appendix.

\subsection{\label{subsec:Example Cumulative-hysteresis}Cumulative hysteresis}

Consider the hysteresis functional $h_{t}(c)=c_{t}\int_{0}^{t}c_{s}ds$.
Here, the path-dependence enters cumulatively as the integral of the
past realization. The cumulative hysteresis is then multiplicative
with the current policy $c_{t}.$ 

We start with the Fréchet derivative of $h_{t}$ that appears in the
first order conditions of the problem (\ref{eq: FOC, perturbed}).
It is given by varying the whole path $c_{[0,t]}$ by a variation
$z_{[0,t]}$:
\[
h_{t}(c+z)-h_{t}(c)=\left(c_{t}+z_{t}\right)\int_{0}^{t}\left(c_{s}+z_{s}\right)ds-c_{t}\int_{0}^{t}c_{s}ds=\int_{0}^{t}c_{t}z_{s}ds+z_{t}\int_{0}^{t}c_{s}ds+o(||z||).
\]
Correspondingly, the present marginal effects of policy are given
by the integral of the path of previous policies
\[
\partial_{c_{t}}h_{t}(c)=\int_{0}^{t}c_{s}ds.
\]
As all previous $c_{s}$ enter identically in the integral, the effects
of the past marginally contributes $c_{t}$:
\[
\delta_{s}h_{t}=c_{t}\mbox{ for }s<t.
\]
 In particular, we see that $h_{t}\in\mathcal{A}_{t}$ and $\delta_{s}h_{t}\in\mathcal{A}_{t}.$

The first order process is given by (\ref{eq: First order process}),
recalling that $c_{t}^{*}=w_{t}$:
\[
C_{t}=-\partial_{c_{t}}h_{t}(c_{[0,t]}^{*})-E\left[\int_{t}^{T}\delta_{t}h_{s}(c_{[0,s]}^{*})ds\bigg|\mathcal{F}_{t}\right]=-\underset{\text{Present effect}}{\underbrace{\int_{0}^{t}w_{s}ds}}-\underset{\text{Future effect}}{\underbrace{E\left[\int_{t}^{T}w_{s}ds\bigg|\mathcal{F}_{t}\right]}}.
\]
Note that in the future effects it is $\delta_{t}h_{s}(c^{t})=c_{s}^{*}=w_{s}$
that enters as it shows how policy in period $t$ ($t\le s\le T)$
affects future period $s$.

We now turn to calculation of the derivatives needed for the application
of the Dupire's functional Itô formula (\ref{eq: FOP, Dupire part}).
Recall that the horizontal and two vertical derivatives are needed.
The horizontal derivative of the present effects of the path $\partial_{c_{t}}h_{t}(c)$
is simply the time derivative of the integral: 
\[
\Delta_{t}\left(\partial_{c_{t}}h_{t}(c)\right)=\Delta_{t}\left(\int_{0}^{t}c_{s}^{*}ds\right)=c_{t}^{*}=w_{t}.
\]
The vertical differentiation of the present effects of the path $\partial_{c_{t}}h_{t}(c)$
is given by bumping the path and is equal to zero:
\begin{align*}
 & \partial_{c_{t}}\left(\partial_{c_{t}}h_{t}(c)\right)=\partial_{c_{t}}\left(\int_{0}^{t}c_{s}^{*}ds\right)=0,\\
 & \partial_{c_{t}}^{3}h_{t}(c)=0.
\end{align*}

We now turn to characterizing the expected future effects of policy
using the total derivative formula (\ref{eq:total derivative formula}).
Since the marginal effect of the policy at time $t$ on time $s$,
$\delta_{t}h_{s}\left(w^{s}\right)=w_{s}$, does not depend on time,
we get
\[
\partial_{t}\left(\delta_{t}h_{s}(w^{s})\right)=0.
\]
The Malliavin derivative is also simple as it measures the sensitivity
of $w^{s}$ to the shock $w_{t}:$
\[
D_{t}\left(\delta_{t}h_{s}\left(w^{s}\right)\right)=D_{t}w_{s}=1,
\]
and the the stochastic derivative term of the future marginal costs
becomes
\[
E\left[D_{t}\int_{t}^{T}\delta_{t}h_{s}(w^{s})ds\bigg|\mathcal{F}_{t}\right]=T-t.
\]
This has a natural interpretation. The stochastic derivative measures
the impact of the stochastic shock $dw_{t}$ on the future marginal
effects of policies. Those are represented by the integrals $\int_{t}^{T}w_{s}ds$.
A shock at time $t$ affects the future marginal effect of policy
in each period $s$as $D_{t}w_{s}=1$. Since there are $\left(T-t\right)$
future periods, the shock $dw_{t}$ has the effect $\left(T-t\right)$.
In other words, the stochastic shocks for early periods $t$ have
longer lasting impact than for the later periods.

We then gather the present and the future terms to determine the dynamics
of the first order process in (\ref{eq: FOP, main formula}):
\begin{align*}
dC_{t}= & -\Bigl(\underset{=w_{t}}{\underbrace{\Delta_{t}\partial_{c_{t}}h_{t}(w^{t})}dt}+\frac{1}{2}\underset{=0}{\underbrace{\partial_{c_{t}}^{3}h_{t}(w^{t})}}dt+\underset{=0}{\underbrace{\partial_{c_{t}}^{2}h_{t}(w^{t})}}dw_{t}\Bigr)-\\
 & -\Bigl(\underset{=w_{t}}{\underbrace{-\delta_{t}h_{t}(w^{t})}dt+}\underset{=0}{\underbrace{E\left[\int_{t}^{T}\partial_{t}\delta_{t}h_{s}(w^{s})ds\bigg|\mathcal{F}_{t}\right]}dt+}\underset{=T-t}{\underbrace{E\left[\int_{t}^{T}D_{t}\delta_{t}h_{s}(w^{s})ds\bigg|\mathcal{F}_{t}\right]}dw_{t}}\Bigr),
\end{align*}

or
\[
dC_{t}=-\left(T-t\right)dw_{t}.
\]

\begin{rem}
\label{rem:Boulatov}There are two alternative ways, without using
our methodology, to derive the result. The first is to note that $E\left[\int_{t}^{T}w_{s}ds\bigg|\mathcal{F}_{t}\right]=\left(T-t\right)w_{t}$
and then
\[
dC_{t}=-w_{t}dt-\left(\left(T-t\right)dw_{t}-w_{t}dt\right)=-\left(T-t\right)dw_{t}.
\]
The second approach that works in a variety of other circumstances
is in extracting a martingale and using the Clark-Ocone formula to
provide its explicit characterization in terms of Malliavin derivatives.
We can rewrite
\[
E\left[\int_{t}^{T}w_{s}ds\bigg|\mathcal{F}_{t}\right]=E\left[\int_{0}^{T}w_{s}ds\bigg|\mathcal{F}_{t}\right]-\int_{0}^{t}w_{s}ds,
\]
and 
\[
C_{t}=-\left(\int_{0}^{t}w_{s}ds+E\left[\int_{0}^{T}w_{s}ds\bigg|\mathcal{F}_{t}\right]-\int_{0}^{t}w_{s}ds\right)=-E\left[\int_{0}^{T}w_{s}ds\bigg|\mathcal{F}_{t}\right].
\]
The Clark-Ocone formula gives the representations of the martingale
where $M_{t}=E\left[\int_{0}^{T}w_{s}ds\bigg|\mathcal{F}_{t}\right]$
as
\[
dM_{t}=E\left[D_{t}\int_{0}^{T}w_{s}ds\bigg|\mathcal{F}_{t}\right]dw_{t}=E\left[\int_{t}^{T}\underset{=1}{\underbrace{D_{t}w_{s}}}ds\bigg|\mathcal{F}_{t}\right]dw_{t}=\left(T-t\right)dw_{t}.
\]
We expand on this class of examples in Boulatov, Riabov, and Tsyvinski
(2020) where we analyze a more general class of environments which
can be solved by extracting a martingale and not using the main tools
of this paper -- the functional Itô formula and the total derivative
formula.
\end{rem}

\subsection{\label{subsec:Example Martingale}Hysteresis with time and past dependency}

Consider the functional $h_{t}(c^{t})=h(c_{t},\int_{0}^{t}a_{t,s}c_{s}ds),$
it describes the interaction of the policy at the present moment of
time, and average of past values. The averaging is defined by the
smooth kernel $a_{t,s}$ which measures how policy in period $s$
affects period $t$. For example, $a_{t,s}=\frac{1}{t}$ corresponds
to the usual averaging. Importantly, the kernel depends on current
time $t$.

This example captures three important parts of history-dependent policies.
First, there is an effect of the present which is represented by the
first argument, $c_{t}$. Second, there is an effect of the past which
is represented by the second argument, $\int_{0}^{t}a_{t,s}c_{s}ds$.
Here, the past enters as the integral of the path of the previous
consumptions where the path enters through the kernel $a_{t,s}$ that
depends both on the current policy $t$ and the path of previous policies.\footnote{One can easily modify this part to be some more complicated path-dependent
object -- for example, $\int_{0}^{t}\alpha\left(c_{\left[0,s\right]}\right)ds$,
where $\alpha$ is a functional of the path or to have arbitrary interaction
of the past effects; or allow for $h_{t}\left(c\right)=\int_{0}^{t}\int_{0}^{t}g_{t}(c_{s},c_{r})dsdr),...$
where now there is a joint dependence on the past values $c_{s}$
and $c_{r}$ via a repeated integral.} Third, there is a joint dependence $h\left(.,.\right)$ between the
present and the past.

The Fréchet derivative of $h_{t}$ that appears in the first order
conditions (\ref{eq: FOC, perturbed}) is given by changing the whole
path $c^{t}$ by a variation $z^{t}$:
\begin{align*}
 & h_{t}(c+z)-h_{t}(c)=h\left(c_{t}+z_{t},\int_{0}^{t}a_{t,s}\left(c_{s}+z_{s}\right)ds\right)-h_{t}(c)=\\
 & =z_{t}h'_{1}\left(c_{t},\int_{0}^{t}a_{t,s}c_{s}ds\right)+\int_{0}^{t}a_{t,s}h'_{2}\left(c_{t},\int_{0}^{t}a_{t,r}c_{r}dr\right)z_{s}ds+o(||z||).
\end{align*}
Correspondingly, the present marginal effects of policy are given
by the integral of the path of previous policies
\[
\partial_{c_{t}}h_{t}(c)=h'_{1}\left(c_{t},\int_{0}^{t}a_{t,s}c_{s}ds\right).
\]
The marginal effect of policy in time $t$ n the past period $s$
($s\le t\le T)$ is given by:
\[
\delta_{s}h_{t}(c^{t})=a_{t,s}h'_{2}\left(c_{t},\int_{0}^{t}a_{t,r}c_{r}dr\right).
\]
 In particular, we see that $h_{t}\in\mathcal{A}_{t}$ and $\delta_{s}h_{t}\in\mathcal{A}_{t}.$

The first order condition (\ref{eq: FOC, perturbed}) is 
\[
w_{t}=c_{t}^{(\epsilon)}+\epsilon h'_{1}\left(c_{t},\int_{0}^{t}a_{t,s}c_{s}ds\right)+\epsilon E\left[\int_{t}^{T}a_{s,t}h'_{2}\left(c_{s},\int_{0}^{s}a_{s,r}c_{r}dr\right)ds\bigg|\mathcal{F}_{t}\right];
\]
and the stochastic elasticity (\ref{eq: First order process}) is
given by, recalling that $c_{t}^{*}=w_{t}$:
\[
C_{t}=-\underset{\text{Present effect}}{\underbrace{h'_{1}\left(w_{t},\int_{0}^{t}a_{t,s}w_{s}ds\right)}}-\underset{\text{Future effect}}{\underbrace{E\left[\int_{t}^{T}a_{s,t}h'_{2}\left(w_{s},\int_{0}^{s}a_{s,r}w_{r}dr\right)ds\bigg|\mathcal{F}_{t}\right]}}.
\]
Note that in the future effects it is $\delta_{t}h_{s}(c^{t})=a_{s,t}h'_{2}\left(w_{s},\int_{0}^{s}a_{s,r}w_{r}dr\right)$
that enters as the kernel $a_{s,t}$ measures how policy in period
$t$ ($t\le s\le T)$ affects future period $s$.

The conditional expectation $E\left[\int_{t}^{T}a_{s,t}h'_{2}\left(w_{s},\int_{0}^{s}a_{s,r}w_{r}dr\right)ds\bigg|\mathcal{F}_{t}\right]$
is intractable for direct computations, as it requires infinite-dimensional
integration over the distribution of the process $(w(r)-w(t))_{t\leq r\leq s}.$
However, the semimartingale decomposition is obtained immediately
after an application of the total derivative formula (\ref{eq:total derivative formula}). 

We now turn to calculation of the derivatives needed for the application
of the Dupire's functional Itô formula (\ref{eq: FOP, Dupire part}).
These derivatives are used in determining the decomposition of the
present effects of the path $\partial_{c_{t}}h_{t}(c)$ in the first
order conditions of the perturbed problem (\ref{eq: FOC, perturbed}).
Recall that the horizontal and two vertical derivatives are needed.

The horizontal derivative of the present effects of the path $\partial_{c_{t}}h_{t}(c)$
is given by freezing the path at time $t$ and extending it with time:
\[
\Delta_{t}\partial_{c_{t}}h_{t}(c)=\Delta_{t}h'_{1}\left(w_{t},\int_{0}^{t}a_{t,s}w_{s}ds\right)=h''_{12}\left(w_{t},\int_{0}^{t}a_{t,s}w_{s}ds\right)\left(a_{t,t}w_{t}+\int_{0}^{t}(\partial_{t}a_{t,s})w_{s}ds\right).
\]
Vertical differentiation of the present effects of the path $\partial_{c_{t}}h_{t}(c)$
is given by bumping the path. It reduces to the differentiation in
the first argument: 
\begin{align*}
 & \partial_{c_{t}}^{2}h_{t}(c)=\partial_{c_{t}}\left(h'_{1}\left(w_{t},\int_{0}^{t}a_{t,s}w_{s}ds\right)\right)=h''_{11}\left(w_{t},\int_{0}^{t}a_{t,s}w_{s}ds\right),\\
 & \partial_{c_{t}}^{3}h_{t}(c)=\partial_{c_{t}}^{2}\left(h'_{1}\left(w_{t},\int_{0}^{t}a_{t,s}w_{s}ds\right)\right)=h'''_{111}\left(w_{t},\int_{0}^{t}a_{t,s}w_{s}ds\right).
\end{align*}
 Note that since we are already finding the derivatives of the marginal
effect $\partial_{c_{t}}h_{t}(c)$, the mixed time-derivative and
the second and third order space derivatives appear (rather than just
the time and first and second order derivatives).

We next calculate the derivatives needed for the application of the
total derivative formula (\ref{eq: FOP, total derivative part}) to
characterize the conditional expectation of the marginal effects of
policy. The time derivative is given by
\[
\partial_{t}\left(\delta_{t}h_{s}(w^{s})\right)=\partial_{t}\left(a_{s,t}h'_{2}\left(w_{s},\int_{0}^{s}a_{s,r}w_{r}dr\right)\right)=(\partial_{t}a_{s,t})h'_{2}\left(w_{s},\int_{0}^{s}a_{s,r}w_{r}dr\right).
\]
The Malliavin derivative is given by
\begin{align*}
 & D_{t}\left(\delta_{t}h_{s}(w^{s})\right)=D_{t}\left(a_{s,t}h'_{2}\left(w_{s},\int_{0}^{s}a_{s,r}w_{r}dr\right)\right)=\\
 & =a_{s,t}h''_{12}\left(w_{s},\int_{0}^{s}a_{s,r}w_{r}dr\right)\underset{=1}{\underbrace{D_{t}w_{s}}}+a_{s,t}h''_{22}\left(w_{s},\int_{0}^{s}a_{s,r}w_{r}dr\right)\underset{=\int_{t}^{s}a_{s,r}D_{t}w_{r}dr=\int_{t}^{s}a_{s,r}dr}{\underbrace{D_{t}\left(\int_{0}^{s}a_{s,r}w_{r}dr\right)}}.
\end{align*}
Gathering the terms we find the differential of the first order process:
\begin{align*}
dC_{t} & =-h''_{12}\left(w_{t},\int_{0}^{t}a_{t,s}w_{s}ds\right)\left(a_{t,t}w_{t}+\int_{0}^{t}(\partial_{t}a_{t,s})w_{s}ds\right)dt-\\
 & -\left(h''_{11}\left(w_{t},\int_{0}^{t}a_{t,s}w_{s}ds\right)dw_{t}+\frac{1}{2}h'''_{111}\left(w_{t},\int_{0}^{t}a_{t,s}w_{s}ds\right)dt\right)-\\
 & -\Bigl(\underset{\delta_{t}h_{t}(w^{t})}{-\underbrace{a_{t,t}h'_{2}\left(w_{t},\int_{0}^{t}a_{t,r}w_{r}dr\right)}}dt+\underset{E\left[\int_{t}^{T}\partial_{t}\delta_{t}h_{s}(w^{s})ds\bigg|\mathcal{F}_{t}\right]}{\underbrace{E\left[\int_{t}^{T}(\partial_{t}a_{s,t})h'_{2}\left(w_{s},\int_{0}^{s}a_{s,r}w_{r}dr\right)ds\bigg|\mathcal{F}_{t}\right]}}dt+\\
 & +\underset{E\left[D_{t}\int_{t}^{T}\delta_{t}h_{s}(w^{s})ds\bigg|\mathcal{F}_{t}\right]}{\underbrace{E\left[\int_{t}^{T}(a_{s,t}h''_{12}\left(w_{s},\int_{0}^{s}a_{s,r}w_{r}dr\right)+a_{s,t}h''_{22}\left(w_{s},\int_{0}^{s}a_{s,r}w_{r}dr\right)\int_{t}^{s}a_{s,r}dr)ds\bigg|\mathcal{F}_{t}\right]}}dw_{t}.
\end{align*}

\begin{rem}
\label{rem:Detemple Zapatero}The presented calculation can be done
without the use of the total derivative formula when the occurrence
of $t$ under the integral in the conditional expectation can be removed,
e.g. when $a_{s,t}=\tilde{g}(t)g(s),$ and a martingale can be extracted.
Indeed, then future effects can be represented as 
\begin{align*}
\tilde{g}(t)E\left[\int_{t}^{T}g(s)h'_{2}\left(w_{s},\int_{0}^{s}a_{s,r}w_{r}dr\right)ds\bigg|\mathcal{F}_{t}\right] & =-\tilde{g}(t)\int_{0}^{t}g(s)h'_{2}\left(w_{s},\int_{0}^{s}a_{s,r}w_{r}dr\right)ds+\tilde{g}(t)M_{t},
\end{align*}
where
\[
M_{t}=E\left[\int_{0}^{T}g(s)h'_{2}\left(w_{s},\int_{0}^{s}a_{s,r}w_{r}dr\right)ds\bigg|\mathcal{F}_{t}\right].
\]
Now the semimartingale decomposition immediately follows from the
Itô formula. This type of examples (with exponential functions $g,\tilde{g}$)
where studied by Detemple and Zapatero (1991) in the context of asset
prices under habit formation.
\end{rem}

\subsection{\label{subsec: Example Deterministic-hysteresis}Deterministic hysteresis}

In this section, we consider a deterministic case. Assume that the
underlying process shock $\theta$ is deterministic: $d\theta_{t}=b(\theta_{t})dt.$
Then optimal policies $c^{*}$ and $c^{\epsilon}$ are deterministic
as well: 
\[
c_{t}^{*}=\theta_{t}.
\]
 Our assumption of $h_{t}$ belonging to the class $\mathcal{A}_{t}$
allows us to neatly decompose its effects into those of the effects
of the present and the past history
\[
c_{t}^{\epsilon}=\theta_{t}-\epsilon\left(\partial_{c_{t}}h_{t}(c^{\epsilon})+\epsilon\int_{t}^{1}\delta_{t}h_{s}(c^{\epsilon})ds\right).
\]
Differentiating the latter relation in $\epsilon$ at $\epsilon=0$
we get 
\[
C_{t}=-\left(\partial_{c_{t}}h_{t}(c^{*})+\int_{t}^{1}\delta_{t}h_{s}(c^{*})ds\right).
\]
Now, apply the functional Itô formula (\ref{eq:Functional Ito}) to
$\partial_{c_{t}}h_{t}(c^{*})$ and note that there are no second
order vertical derivatives of $\partial_{c_{t}}h_{t}(c^{*})$ due
to the absence of stochasticity:
\[
d\left(\partial_{c_{t}}h_{t}(c^{*})\right)=\left(\Delta_{t}\partial_{c_{t}}h_{t}(c^{*})+\partial_{c_{t}}^{2}h_{t}(c^{*})b(c_{t}^{*})\right)dt.
\]
The future marginal effects of policy are also deterministic and we
can simply differentiate $\int_{t}^{1}\delta_{t}h_{s}(c^{*})ds$ in
time.

\[
dC_{t}=-\bigg(\underset{d\left(\partial_{c_{t}}h_{t}(c^{*})\right)}{\underbrace{\Delta_{t}\partial_{c_{t}}h_{t}(c^{*})+\partial_{c_{t}}^{2}h_{t}(c^{*})b(c_{t}^{*})}}\bigg)dt-\bigg(\underset{d\left(\int_{t}^{1}\delta_{t}h_{s}(c^{*})ds)\right)}{\underbrace{-\delta_{t}h_{t}(c^{*})+\int_{t}^{1}\frac{\partial}{\partial t}\delta_{t}h_{s}(c^{*})ds)}}\bigg)dt.
\]

The present effects of the policy $\partial_{c_{t}}h_{t}(c^{*})$
change with time as the horizontal derivative $\Delta_{t}$ and with
the movement in the path of the policy as the vertical derivative
$\partial_{c_{t}}$. The cumulative future marginal effects of the
policy $\int_{t}^{1}\delta_{t}h_{s}(c^{*})ds$ change as the relative
difference between the time change in the cumulant of the future marginal
effects $\frac{\partial}{\partial t}\delta_{t}h_{s}(c^{*})ds$ relative
to today's marginal effect of the policy $\delta_{t}h_{t}(c^{*})$.
This case is interesting as a stand alone result that allows to focus
on history dependence without stochasticity.

\subsection{\label{subsec:A-tipping-point}A ``tipping point''}

The main premise of the literature on tipping points discussed in
\ref{subsec:Evidence-on-climate} is that there is some level of emissions
the crossing of which leads to a different behavior of the climate
system. For example, upon reaching the threshold the damages become
larger or become irreversible. We now show an example where instead
of considering a threshold we focus on the time when the climate variable
achieves its maximum upto any given period of time. 

Let $\theta_{t}$ be the time when the Brownian motion $w$ achieves
its maximum over $[0,t]:$
\[
\theta_{t}=\arg\max_{s\in[0,t]}w_{s}
\]
and it is known that $\theta_{t}$ is a.s. unique. Let $f(x)$ be
an absolutely continuous function such that $f(x)=0$ for $x\leq0.$
Consider the objective function
\[
E\left[\int_{0}^{T}\left(-\frac{1}{2}(c_{t}-w_{t})^{2}-\int_{0}^{t}f(s-\theta_{t})c_{s}ds\right)dt\right].
\]
One can think of this setting as follow. In each period $t$, we consider
the time $\theta_{t}$ when the maximal level of $w$ has been achieved
-- say, the time of the temperature record. The damages then are
counted as the weighted by $f\left(s-\theta_{t}\right)$ emissions
from the time of the record. When the new record is achieved, the
weighting restarts. This example can be significantly expanded by
having the weighting functions change with time or the weighting functions
that weigh both the time prior to $\theta_{t}$ and the time after
$\theta_{t}$ (a form of highlighting the ``salience'' of the record
time) but we chose to present the simple form here. That is, there
is a salient ``tipping'' or reference point following which the
damages change their behavior.

The first-order conditions give us the optimal policy in a closed
form is given by
\[
c_{t}=w_{t}-E\left[\int_{t}^{T}f(t-\theta_{s})ds\bigg|\mathcal{F}_{t}\right].
\]
The process $\xi_{t}=\int_{t}^{T}f(t-\theta_{s})ds$ is absolutely
continuous and square integrable. Indeed,
\[
\xi_{t}=\int_{t}^{T}\int_{0}^{t}f'(x-\theta_{s})dxds.
\]
Let us represent $\xi_{t}$ using the Clark-Ocone formula
\[
\xi_{t}=\mathbb{E}\xi_{t}+\int_{0}^{T}g_{t,s}dw_{s}.
\]
By the total derivative formula of Proposition \ref{prop: Total derivative formula},
we immediately find that $c$ is an Itô process and 
\begin{equation}
dc_{t}=\bigg(1-g_{t,t}\bigg)dw_{t}-\mathbb{E}[\partial_{t}\xi_{t}|\mathcal{F}_{t}]dt.\label{eq: Tipping point}
\end{equation}

\begin{rem}
Let us try to get this result using standard methods. We compute the
process $c_{t}$ explicitly as a functional of the Wiener process
$w.$ At first we find the conditional expectation 
\begin{align*}
 & E[f(t-\theta_{s})|\mathcal{F}_{t}]=E\left[\int_{0}^{\infty}f'(x)1_{x<t-\theta_{s}}dx\bigg|\mathcal{F}_{t}\right]=\\
 & =\int_{0}^{t}f'(x)P(\theta_{s}<t-x|\mathcal{F}_{t})dx=\int_{0}^{t}f'(x)P\left(\max_{[0,t-x]}w>\max_{[t-x,s]}w\bigg|\mathcal{F}_{t}\right)dx=\\
 & =\int_{0}^{t}f'(x)1_{\max_{[0,t-x]}w>\max_{[t-x,t]}w}P\left(\max_{[0,t-s]}w<z\right)\bigg|_{z=\max_{[0,t-x]}w-w_{t}}dx=\\
 & =\int_{0}^{t}f'(x)1_{\max_{[0,t-x]}w>\max_{[t-x,t]}w}\left(2\Phi\left(\frac{\max_{[0,t-x]}w-w_{t}}{\sqrt{t-s}}\right)-1\right)dx.
\end{align*}
Hence,
\begin{align*}
c_{t} & =w_{t}-\int_{0}^{t}f'(x)1_{\max_{[0,t-x]}w>\max_{[t-x,t]}w}\int_{t}^{T}\left(2\Phi\left(\frac{\max_{[0,t-x]}w-w_{t}}{\sqrt{t-s}}\right)-1\right)dsdx=\\
 & =w_{t}-\int_{0}^{t-\theta_{t}}f'(x)\int_{t}^{T}\left(2\Phi\left(\frac{\max_{[0,t]}w-w_{t}}{\sqrt{t-s}}\right)-1\right)dsdx=\\
 & =w_{t}-f(t-\theta_{t})\int_{t}^{T}\left(2\Phi\left(\frac{\max_{[0,t]}w-w_{t}}{\sqrt{t-s}}\right)-1\right)ds.
\end{align*}
We get the representation of the form
\[
c_{t}=w_{t}-f(t-\theta_{t})g(\max_{[0,t]}w-w_{t}).
\]
Now one can try two standard approaches to get the semimartingale
decompostion. The first is to apply the Itô formula. In order to do
this, we see that 
\[
c_{t}=F(t,w_{t},\theta_{t},M_{t}),
\]
where $M_{t}=\max_{[0,t]}w,$ and $F(t,x,y,z)=x-f(t-y)g(z-x).$ Hence
the Itô formula will lead to a semimartingale decomposition that contains
terms $d\theta_{t},$ $dM_{t},$ and the sum over jumps of the process
$\theta_{t}.$ It is not obvious that the process $c_{t}$ is an Itô
process.
\end{rem}
The second approach is to apply the functional Itô formula. In order
to do this the functional $c_{t}$ must be extended to a $C_{b}^{1,2}$-functional
(in the sense of Cont and Fournie) on the space of cadlag paths. However,
if such extension is possible, vertical derivatives of $c_{t}$ must
be equal to zero and the functional Itô formula would lead to semimartingale
decomposition $dc_{t}=\alpha_{t}dt,$ which is not the case.

Instead, the application of our methodology give a straightforward
and compact answer of the equation (\ref{eq: Tipping point}).

\section{Discussion}

In this section, we discuss some of the more technical issues behind
the results.

\subsection{Why do we need both the functional Itô formula and the total derivative
formula?\label{subsec:Why-do-we-need-two-formulas}}

The stochastic elasticity $C_{t}^{h}$ is described in \eqref{eq: First order process}
via two terms: a vertical derivative $\partial_{c_{t}}h_{t}(c)$ of
the perturbation functional $h_{t}$ and the conditional expectation
$\mathbb{E}[\int_{t}^{1}\delta_{t}h_{s}(c^{*})ds|\mathcal{F}_{t}].$
Both are functionals of the path $(c_{s})_{0\leq s\leq t}$ but we
need to use different methods to study them.

We assume that the functional $\partial_{c_{t}}h_{t}(c)$ satisfies
conditions from Dupire (2019). This is a natural assumption, since
$\partial_{c_{t}}h_{t}(c)$ is a well-defined time-dependent functional
of the path $(c_{s})_{0\leq s\leq t}$.

However, in the case of the functional $\mathbb{E}[\int_{t}^{1}\delta_{t}h_{s}(c^{*})ds|\mathcal{F}_{t}]$
the same assumption is not generally suitable. The application of
the Dupire's functional Itô formula requires that 
\[
\mathbb{E}\left[\int_{t}^{1}\delta_{t}h_{s}(c^{*})ds|\mathcal{F}_{t}\right]=q_{t}(m),
\]
where $m$ is a certain semimartingale and $(q_{t})_{t\in[0,T]}$
is smooth in the sense of Dupire (2019) and Cont and Fournie (2013)
family of path-dependent functionals. It may be a difficult stand-alone
problem even to verify that such representation holds. Consider, for
example, the case when 
\[
\mathbb{E}\left[\int_{t}^{1}\delta_{t}h_{s}(c^{*})ds|\mathcal{F}_{t}\right]=\int_{0}^{t}f_{t}(c_{s}^{*})dc_{s}^{*},
\]
and $f_{t}$ is smooth in $t.$ The stochastic integral is defined
for almost all realizations of $c^{*}$ only and it is unclear how
we can extend it smoothly to all continuous paths.\footnote{Recall that the functional It\textroundcap{o} formula is applicable
to functionals defined on a larger space of cádlág paths.}

Instead we assume that the expression under the conditional expectation
is smooth in $t$ and then apply the total derivative formula in Proposition
\ref{prop: Total derivative formula} to prove directly that $C^{h}$
is an Itô process and to find its semimartingale decomposition as
well.

In other words, the functional Itô formula is useful when the functional
of the path is already well-defined. When faced with a conditional
expectation process such as the one we considered here or that frequently
occurs in a variety of other economic problems, the total derivative
formula allows to straightforwardly calculate its semimartingale decomposition.

\subsection{\label{subsec:Assumptions-on-smoothness}Assumptions on smoothness
of the functionals}

Our results are valid for path-dependent functionals $h=(h_{t})_{0\leq t\leq T}$
such that: (1) $h_{t}\in\mathcal{A}_{t};$ (2) $t\to\partial_{c_{t}}h_{t}$
is horizontally and twice vertically differentiable; (3) $t\to\delta_{t}h_{s}(c^{s})$
is absolutely continuous; (4) $w\to\delta_{t}h_{s}(c^{s}(w))$ is
Malliavin differentiable.

The first restriction that $h_{t}$ belongs to the class $\mathcal{A}_{t}$
allows to calculate the first order conditions in a tractable form.
The second restriction is a typical condition needed for the functional
Itô formula to be valid. Last two conditions are imposed for the ease
of presentation and also they are satisfied in all our examples. They
are needed for the total derivative formula to be applicable for the
process 
\[
t\to E[\int_{t}^{T}\delta_{t}h_{s}(c^{s})ds|\mathcal{F}_{t}].
\]
However, these conditions can be considerably relaxed. It is enough
to find an absolutely continuous square integrable process $\xi_{t},$
such that
\[
E[\int_{t}^{T}\delta_{t}h_{s}(c^{s})ds|\mathcal{F}_{t}]=E[\xi_{t}|\mathcal{F}_{t}].
\]
As a simple example consider the process
\[
\delta_{t}g_{s}(c_{s})=w_{t}.
\]
This process is not absolutely continuous, however, $E[w_{t}|\mathcal{F}_{t}]=E[w_{T}|\mathcal{F}_{t}]$
and one can take $\xi_{t}=(T-t)w_{T}.$

\subsection{\label{subsec:Relationship-to-the_Clarrk-Ocone}Relationship to the
Clark-Ocone formula}

It is instructive to compare Proposition \ref{prop: Total derivative formula}
to the well-known Clark-Ocone formula. Recall that $\mathcal{F}=(\mathcal{F}_{t})_{t\in[0,T]}$
is a filtration generated by Wiener process $(w_{t})_{t\in[0,T]}.$
Every $\mathcal{F}$-martingale $(M_{t})_{t\in[0,T]}$ can be represented
as a conditional expectation process. Indeed, from the definition
of a martingale we get $M_{t}=\mathbb{E}[M_{T}|\mathcal{F}_{t}]$.
Conversely, every process of the form 
\[
Z_{t}=\mathbb{E}[\xi|\mathcal{F}_{t}],
\]
where $\xi$ is an integrable random variable, is a martingale. This
follows from the basic properties of conditional expectations:

\[
\mathbb{E}\left[Z_{t+h}|\mathcal{F}_{t}\right]=\mathbb{E}\left[\mathbb{E}\left[\xi|\mathcal{F}_{t+h}\right]|\mathcal{F}_{t}\right]=\mathbb{E}\left[\xi|\mathcal{F}_{t}\right]=Z_{t}.
\]
Moreover, every $\mathcal{F}$-martingale $Z_{t}=\mathbb{E}[\xi|\mathcal{F}_{t}]$
is an Itô process and if $\xi$ is Malliavin differentiable, the Clark-Ocone
formula gives 
\[
Z_{t}=Z_{0}+\int_{0}^{t}\mathbb{E}[D_{s}\xi|\mathcal{F}_{s}]dw_{s},
\]
and the differential of $Z:$

\[
dZ_{t}=\mathbb{E}[D_{t}\xi|\mathcal{F}_{t}]dw_{t}.
\]

The processes that we consider in the Proposition \ref{prop: Total derivative formula}
are of the different type - they are of the form 
\[
X_{t}=\mathbb{E}[\xi_{t}|\mathcal{F}_{t}]
\]
 with $\xi=(\xi_{t})_{t\in[0,T]}$ being an absolutely continuous
process. Such process are not martingales in general. The simplest
example is provided with deterministic non-constant process $\xi$,
e.g. if we take $\xi_{t}=t,$ then $X_{t}=t$ and this is obviously
not a martingale.

The Clark-Ocone formula is then a partial case of Proposition \ref{prop: Total derivative formula}
corresponding to constant in $t$ process $(\xi_{t})_{t\in[0,T]}$.
For the proof and general versions of Proposition \ref{prop: Total derivative formula}
we refer to the Appendix.

\subsection{\label{subsec:Discussion-of-the-optimal control}Discussion of the
optimal control approach}

A natural question is whether we can use optimal control and dynamic
programming to simplify the analysis. 

One of the standard approaches for solving optimization problems is
the dynamic programming principle. It can be used for maximization
of functionals of the type (see see Fabbri, Gozzi, and Swiech 2017
for the exposition of infinitely dimensional problems):
\[
V(a)=E\int_{0}^{T}l(t,x^{t},a(t))dt,
\]
where $a$ is the control process, the process $x$ satisfies certain
SDE whose coefficients depend on the control $a$ as well. This framework
is not well suited for our problem, as in our case the dependence
of the expression under the integral is a path-dependent functional
of the policy $c$: 
\[
V(c)=E\int_{0}^{T}\left(-\frac{1}{2}(c_{t}-w_{t})^{2}-h_{t}(c^{t})\right)dt.
\]
Hence, to fit the framework of the optimal control we must consider
the whole path $c^{t}$ as the value of the control $a(t).$ That
is, the control space becomes the set of paths $\cup_{t\in[0,T]}C[0,t].$
This leads to the consideration of the cost functional
\[
V(t,x,c^{t})=\max_{a\in C[t,T]}E_{t,x}\int_{t}^{T}\left(-\frac{1}{2}(a_{s}-w_{s-t}-x)^{2}-h_{s}(c^{t}\otimes a^{s})\right)ds.
\]
The HJB equation then becomes a path-dependent PDE that contains infinite-dimensional
optimization over trajectories of the control. Thus the analysis may
be as complicated as our original problem. For example, Cosso et.
al (2018) use functional Itô calculus to write the path-dependent
HJB equation, prove that the value function is a viscosity solution
and prove a partial comparison principle. 

\section{Conclusion}

Path-dependent policies are a feature of a number of economic models.
Doing a lot of sports only when young leads to a very different health
outcome than exercising lightly throughout life, even though the cumulative
lifetime exercise may be the same. Moreover, the importance of the
previous pattern of exercise may change with time. For example, intensively
training several times a week is best when young while a steady pattern
of light exercise during a week may be preferable when older. Missing
a credit card payment twice in a row in a ten year period leads to
a very different credit history than missing a payment once every
five years. The consequence of the same sequence of non-payments for
the credit history may be quite different in a recession versus a
boom. These policies are not just a function of the current state
or time but depend on the whole trajectory of the past actions. Moreover,
the dependency itself is changing with time. 

We developed a methodology for analysis of general class of policy
with path-dependent effects in an uncertain environment with forward
looking agents. The primary difficulty that arises in such models
is that the optimizing agents foresee such path dependency and the
actions they take incorporate the expectation of the future effects.
We show that three ingredients are needed for the analysis of such
problems. First, we introduced a general class of path-dependent functionals
that allows to tractably write the first order conditions for the
problem. Second, the recently developed functional Itô calculus allows
to describe the dynamics of the present effect of the past choices.
Third, the total derivative formula that we develop in this paper
allows to derive the dynamics of the conditional expectation processes
of the future effects of today's choices. Our analysis shows that
even when policy has very small contemporaneous effects, it may have
large effects due to either interaction with the past choices or due
to the expected future effects. The effects of the past are represented
by the magnitude of horizontal derivatives and higher order vertical
derivative. The future expected effects are determined by the magnitude
of the time derivative and by the stochastic derivatives. We believe
that the methodology we develop in, on purpose, a stark underlying
environment will facilitate analysis of a wide class of other economic
problems. 

\pagebreak{}

\section{Appendix}

\subsection{Remark \ref{rem:Discontinuous} in Section \ref{subsec:Class A_t}}

The objective function (\ref{eq:Perturbed problem}) can be written
as follows:
\begin{align*}
 & \mathbb{E}\int_{0}^{T}(-\frac{1}{2}\left(c_{t}-\theta_{t}\right)^{2}-\epsilon c_{\frac{t}{2}})dt=\mathbb{E}\int_{0}^{T}(-\frac{1}{2}\left(c_{t}-\theta_{t}\right)^{2})dt-\mathbb{E}\int_{0}^{T}\epsilon c_{\frac{t}{2}}dt=\\
 & =\mathbb{E}\int_{0}^{T}(-\frac{1}{2}\left(c_{t}-\theta_{t}\right)^{2})dt-2\mathbb{E}\int_{0}^{\frac{T}{2}}\epsilon c_{s}ds=\\
 & =\mathbb{E}\int_{0}^{\frac{T}{2}}(-\frac{1}{2}\left(c_{t}-\theta_{t}\right)^{2}-2\epsilon c_{t})dt+\mathbb{E}\int_{\frac{T}{2}}^{T}(-\frac{1}{2}\left(c_{t}-\theta_{t}\right)^{2})dt.
\end{align*}
The first-order conditions become
\[
\begin{cases}
c_{t}=\theta_{t}-2\epsilon,\ t<\frac{T}{2}\\
c_{t}=\theta_{t},\ t>\frac{T}{2},
\end{cases}
\]
which implies that the optimal policy has a jump at $t=\frac{T}{2}$.

\subsection{\label{sec:A-simple-example}An example with contemporaneous policies}

In this section, we develop a simple example of introducing additional
effects of policies that can be easily handled by the usual Itô formula.
That is, there is no hysteresis and policy only has contemporaneous
effects.

Consider an additional effect of policies given by $f(c_{t})$, where
$f:\mathbb{R}\to\mathbb{R}$ is a smooth function:
\[
\max_{c}\mathbb{E}\int_{0}^{T}(-\frac{1}{2}\left(c_{t}-w_{t}\right)^{2}-\epsilon f(c_{t}))dt,
\]
and $\epsilon$ is a parameter. The first order conditions are given
by
\begin{equation}
c_{t}^{\epsilon}=w_{t}-\epsilon f'\left(c_{t}^{\epsilon}\right),\label{eq:FOC simple example}
\end{equation}
and equate the marginal benefit of the policy tracking the process
$w_{t}$ with the additional marginal effects $\epsilon f'\left(c_{t}^{\epsilon}\right)$
.

\subsubsection{Stochastic elasticity and its dynamics}

We are particularly interested in the case of $\epsilon\rightarrow0$.
One can think about this as a small parameter asymptotics for the
problem. Let $C_{t}^{f}$ be a first variation process:
\[
c_{t}^{\epsilon}=c_{t}^{*}+\epsilon C_{t}^{f}+o(\epsilon).
\]
The process $C_{t}^{f}$ represents how optimal policy process changes
locally in response to the introduction of a small effect of policy$f$.
One can think of $C_{t}^{f}$ as a notion of stochastic elasticity
that represents the change in the whole process when policy has additional
effects $f$.

Differentiating the first order condition (\ref{eq:FOC simple example})
with respect to $\epsilon$ and evaluating at $\epsilon=0$, we get
\[
C_{t}^{f}=f'(c_{t}^{*}).
\]
Without any restrictions on $f$, the process $C_{t}^{f}$ may be
quite arbitrary, for example, discontinuous. However, since we assumed
smoothness of $f$, Itô's lemma implies that the process $C_{t}^{f}$
is a semimartingale. Moreover, its decomposition is given as (using
also that $c_{t}^{*}=w_{t}$) in closed form:
\[
dC_{t}^{f}=\alpha(w_{t})dt+\beta(w_{t})dw_{t},
\]
where
\begin{align}
\alpha(w_{t}) & =\frac{1}{2}f'''(w_{t}),\nonumber \\
\beta(w_{t}) & =f''(w_{t}).\label{eq: FOP state dependent}
\end{align}

\subsubsection{Dynamics of optimal policy}

We now turn to characterizing the optimal policy $c_{t}^{\epsilon}$
rather than the first variation process. Let the process $c_{t}^{\epsilon}$
have the form
\[
dc_{t}^{\epsilon}=\alpha^{\epsilon}\left(c_{t}^{\epsilon}\right)dt+\beta^{\epsilon}\left(c_{t}^{\epsilon}\right)dw_{t}.
\]
Applying Itô formula to (\ref{eq:FOC simple example}), we can find
the dynamics of the process $c_{t}^{\epsilon}$:
\begin{align*}
 & dc_{t}^{\epsilon}=dw_{t}-\epsilon\left(f''\left(c_{t}^{\epsilon}\right)dc_{t}^{\epsilon}+\frac{1}{2}f'''\left(c_{t}^{\epsilon}\right)\left(\beta_{t}^{\epsilon}\right)^{2}dt\right),\\
 & \alpha_{t}^{\epsilon}dt+\beta_{t}^{\epsilon}dw_{t}=dw_{t}-\epsilon\left(f''\left(c_{t}^{\epsilon}\right)\alpha_{t}^{\epsilon}+\frac{1}{2}f'''\left(c_{t}^{\epsilon}\right)\left(\beta_{t}^{\epsilon}\right)^{2}\right)dt-\epsilon f''\left(c_{t}^{\epsilon}\right)\beta_{t}^{\epsilon}dw_{t}.
\end{align*}
The drift and diffusion coefficients are then given by collecting
terms at $dt$ and $dw_{t}$:
\begin{align}
\alpha_{t}^{\epsilon}\left(c_{t}^{\epsilon}\right) & =-\epsilon\frac{\frac{1}{2}f'''\left(c_{t}^{\epsilon}\right)\left(\beta_{t}^{\epsilon}\right)^{2}}{1+\epsilon f''\left(c_{t}^{\epsilon}\right)},\nonumber \\
\beta_{t}^{\epsilon}\left(c_{t}^{\epsilon}\right) & =\frac{1}{1+\epsilon f''\left(c_{t}^{\epsilon}\right)}.\label{eq: Coefficients state dependent}
\end{align}
The coefficients in (\ref{eq: Coefficients state dependent}) are
given as a system of coupled equations and themselves depend on the
process $c_{t}^{\epsilon}$. This is in contrast to the coefficients
of the stochastic elasticity given in (\ref{eq: FOP state dependent})
which are given in closed form as they are evaluated at the process
$c_{t}^{*}=c_{t}^{\epsilon=0}=w_{t}.$

The key to tractability in the simple setting of this section is Itô
lemma. However, Itô lemma only applies to functions and not the functionals.
That is, it does not apply to the path-dependent effects of policies
studying which is the main goal of the rest of this paper.

It is useful to summarize these results in terms of how the magnitude
of the additional effects influences policy. This parallels the discussion
of Dixit (1991), Reis (2006) and Alvarez, Lippi, and Paciello (2011,
2016) who show that in the environments with uncertainty and adjustment
frictions the costs up to the fourth order may have first order effects
on the dynamics of optimal policies. In our setting, since the additional
effects $f$ are smooth, we show that the effects up to the third
order, that is the second and the third derivative of $f$ will have
the first-order effects. The first and second derivative of the marginal
effects (that is, the second and third derivative of $f$) matter
for the drift of the optimal policy -- this is the evolution of optimal
policy with respect to time. The first derivative of the marginal
effects matters for the diffusion coefficients -- this is response
of optimal policy to stochastic shocks. Importantly, the dynamics
of the optimal policy is not directly influenced by either the past
or the future evolution of the policies.

\subsection{\label{subsec:Appendix Vertical-derivatives}Vertical derivatives
and the class $\mathcal{A}_{t}$ of functionals}

The functional $h:C[0,t]\to\mathbb{R}$ is said to be in the class
$\mathcal{A}_{t}$ if for each path $c\in C[0,t]$ there exists a
number $\partial_{c_{t}}h(c)$ and an integrable function $(\delta_{s}h(c))_{s\in[0,t]},$
such that the following asymptotic relation holds for any $c\in C[0,t]:$
\[
h(c+z)=h(c)+\partial_{c_{t}}h(c)\cdot z_{t}+\int_{0}^{t}\delta_{s}h(c)\cdot z_{s}ds+o(||z||),\ z\to0.
\]
The number $\partial_{c_{t}}h(c)$ is the derivative of $h$ along
the value of the path $(c_{s})_{s\in[0,t]}$ at time $t$ (i.e. the
derivative along the present value). The function $(\delta_{s}h(c))_{s\in[0,t]},$
represents the integrated influence of the past on the variation of
the functional.
\begin{lem}
Assume that the functional $h\in\mathcal{A}_{t}$ admits a continuous
extension to the space $D[0,t]$ of cadlag paths (equipped with the
Skorokhod topology). Then the functional $h$ is vertically differentiable
at every $c\in C[0,t],$ and the vertical derivative coincides with
$\partial_{c_{t}}h(c).$
\end{lem}
\begin{proof}
Fix $\alpha>0.$ There exists $\delta>0$ such that for any $z\in C[0,t]$
with $||z||\leq\delta,$ 
\[
\bigg|h(c+z)-h(c)-\partial_{c_{t}}h(c)\cdot z_{t}-\int_{0}^{t}\delta_{s}h(c)\cdot z_{s}ds\bigg|\leq\alpha||z||.
\]
The function $e_{t}\in D[0,t]$ can be approximated in the Skorokhod
topology by continuous functions $z^{(n)},$ such that $0\leq z^{(n)}\leq1,$
$z_{s}^{(n)}=0$ for $s\leq t-\frac{1}{n},$ $z_{t}^{(n)}=1.$ We
have for all $\epsilon\leq\delta$ 
\[
\bigg|h(c+\epsilon z^{(n)})-h(c)-\partial_{c_{t}}h(c)\cdot\epsilon-\epsilon\int_{t-\frac{1}{n}}^{t}\delta_{s}h(c)\cdot z_{s}^{(n)}ds\bigg|\leq\alpha\epsilon||z^{(n)}||=\alpha\epsilon.
\]
Dividing by $\epsilon$ we get 
\[
\bigg|\frac{h(c+\epsilon z^{(n)})-h(c)}{\epsilon}-\partial_{c_{t}}h(c)\bigg|\leq\alpha+\int_{t-\frac{1}{n}}^{t}|\delta_{s}h(c)|ds
\]
Taking $n\to\infty,$ 
\[
\bigg|\frac{h(c+\epsilon e_{t})-h(c)}{\epsilon}-\partial_{c_{t}}h(c)\bigg|\leq\alpha.
\]
Since $\alpha>0$ is arbitrary, this proves that the vertical derivative
of $h$ exists and is equal to $\partial_{c_{t}}h(c).$ 
\end{proof}

\subsection{Proof of Claim \ref{claim:Class A is dense}}

We recall that the class $\mathcal{A}_{T}$ consists of functionals
$g:C[0,T]\to\mathbb{R}$ such that for all $c,z\in C[0,T]$ 
\[
g(c+\epsilon z)=g(c)+\epsilon\partial_{c_{t}}g(c)z_{t}+\epsilon\int_{0}^{T}\delta_{s}g(c)z_{s}ds+o(\epsilon),\ \epsilon\to0.
\]
For any function $c\in C[0,T]$ introduce the transformation 
\[
(R^{(n)}c)_{t}=e^{-n(T-t)}c_{T}+n\int_{t}^{T}e^{-n(s-t)}c_{s}ds.
\]
The function $R^{(n)}c$ solves the problem 
\[
\begin{cases}
\frac{dR^{(n)}c}{dt}=-n(c_{t}-(R^{(n)}c)_{t})\\
(R^{(n)}c)_{T}=c_{T}
\end{cases}
\]
Hence, $R^{(n)}c\to c$ uniformly on $[0,T].$ Given any Frechet differentiable
functional $g:C[0,T]\to\mathbb{R},$ consider
\[
g^{(n)}(c)=g(R^{(n)}c).
\]
Fix $c\in C[0,T]$ and let the measure $\mu$ be the Frechet derivative
of $g$ at $R^{(n)}c.$ 
\begin{align*}
g^{(n)}(c+\epsilon z) & =g(R^{(n)}c+\epsilon R^{(n)}z)=g(R^{(n)})(c)+\epsilon\int_{0}^{T}(R^{(n)}z)_{s}\mu(ds)+o(\epsilon)=\\
 & =g^{(n)}(c)+\epsilon\int_{0}^{T}\left(e^{-n(T-s)}z_{T}+n\int_{s}^{T}e^{-n(r-s)}z_{r}dr\right)\mu(ds)+o(\epsilon)=\\
 & =g^{(n)}(c)+\epsilon\left(\int_{0}^{T}e^{-n(T-s)}\mu(ds)\right)z_{T}+\epsilon\int_{0}^{T}\left(\int_{0}^{r}e^{-n(r-s)}\mu(ds)\right)z_{r}dr+o(\epsilon)
\end{align*}
So, $g^{(n)}\in\mathcal{A}_{T}.$ Every differentiable functional
is a pointwise limit of functionals from the class $\mathcal{A}_{T}.$

\subsection{Total derivative formula: a general version of Proposition \ref{prop: Total derivative formula}}

Let $\zeta=(\zeta_{t})_{t\in[0,T]}$ be a measurable process. We will
say that $\zeta$ is square integrable, if $\mathbb{E}\int_{0}^{T}\zeta_{s}^{2}ds<\infty.$

We say that a measurable stochastic process $\xi=(\xi_{t})_{t\in[0,T]}$
is absolutely continuous, if $\xi_{0}$ is square integrable and there
exists a square integrable process $(\zeta_{t})_{t\in[0,T]}$, such
that 
\begin{equation}
\xi_{t}=\xi_{0}+\int_{0}^{t}\zeta_{s}ds,\ 0\leq t\leq T.\label{abs_cont_process}
\end{equation}
Observe that an absolutely continuous process $\xi$ satisfies $\sup_{0\leq t\leq T}\mathbb{E}\xi_{t}^{2}<\infty$
and is a.s. continuous. It is important that an absolutely continuous
process need not be adapted. Since the processes under consideration
are square integrable and measurable with respect to a Wiener process
$w$, their values can be represented as stochastic integrals. Namely,
\begin{equation}
\xi_{t}=\mathbb{E}\xi_{t}+\int_{0}^{T}g_{t,s}dw_{s},\ \zeta_{t}=\mathbb{E}\zeta_{t}+\int_{0}^{T}h_{t,s}dw_{s},\ 0\leq t\leq T,\label{adapted_derivatives}
\end{equation}
where $g_{t}=(g_{t,s})_{s\in[0,T]}$ and $h_{t}=(h_{t,s})_{s\in[0,T]}$
are progressively measurable (in $s$) square integrable processes.
We need one technical result on the regularity of the process $g_{t}.$
\begin{lem*}
For each $t\in[0,T]$ and almost all $s\in[0,T]$ 
\[
g_{t,s}=g_{0,s}+\int_{0}^{t}h_{r,s}dr\mbox{ a.s.}
\]
In particular, for almost all $s\in[0,T]$ the process $t\to g_{t,s}$
has a continuous modification.
\end{lem*}
\begin{proof}
Let $u=(u_{s})_{s\in[0,T]}$ be an arbitrary progressively measurable
square integrable process. Using Itô's isometry we compute 
\begin{align*}
\mathbb{E}\int_{0}^{T}g_{t,s}u_{s}ds & =\mathbb{E}\int_{0}^{T}g_{t,s}dw_{s}\int_{0}^{T}u_{s}dw_{s}=\mathbb{E}\xi_{t}\int_{0}^{T}u_{s}dw_{s}=\\
 & =\mathbb{E}\xi_{0}\int_{0}^{T}u_{s}dw_{s}+\int_{0}^{t}\left(\mathbb{E}\zeta_{r}\int_{0}^{T}u_{s}dw_{s}\right)dr=\\
 & =\mathbb{E}\int_{0}^{T}g_{0,s}u_{s}ds+\int_{0}^{t}\left(\mathbb{E}\int_{0}^{T}h_{r,s}u_{s}ds\right)dr=\mathbb{E}\int_{0}^{T}\left(g_{0,s}+\int_{0}^{t}h_{r,s}dr\right)u_{s}ds.
\end{align*}
Since the latter holds for arbitrary $u$ we deduce that $g_{t}=g_{0}+\int_{0}^{t}h_{r}dr$
as elements of $L^{2}(\Omega\times[0,T]).$
\end{proof}
Further we will always deal with a continuous in t modifications of
processes $t\to g_{t,s}.$
\begin{prop*}
Let $\eta=(\eta_{t})_{t\in[0,T]}$ be a square integrable progressively
measurable process. Then $\eta$ is an Itô process if and only if
it can be represented in the form $\eta_{t}=\mathbb{E}[\xi_{t}|\mathcal{F}_{t}]$
for some square integrable absolutely continuous process $\xi=(\xi_{t})_{t\in[0,T]}.$
In this case the semimartingale representation of $\eta$ is given
by 
\[
d\eta_{t}=\mathbb{E}[\zeta_{t}|\mathcal{F}_{t}]dt+g_{t,t}dw_{t},
\]
where processes $\zeta$ and $g_{t}$ are determined from \eqref{abs_cont_process},
\eqref{adapted_derivatives}.
\end{prop*}
\begin{proof}
Assume that $\eta$ is an Itô process. Then it can be written in the
form
\[
\eta_{t}=\eta_{0}+\int_{0}^{t}\alpha_{s}ds+\int_{0}^{t}\beta_{s}dw_{s},\ 0\leq t\leq T.
\]
Introduce the process
\[
\xi_{t}=\eta_{0}+\int_{0}^{T}\beta_{s}dw_{s}+\int_{0}^{t}\alpha_{s}ds.
\]
Observe that the process $\xi$ is an absolutely continuous process.
Since stochastic integrals are martingales, we deduce that 
\[
\mathbb{E}\left[\xi_{t}|\mathcal{F}_{t}\right]=\eta_{0}+\int_{0}^{t}\alpha_{s}ds+\mathbb{E}\left[\int_{0}^{T}\beta_{s}dw_{s}|\mathcal{F}_{t}\right]=\eta_{0}+\int_{0}^{t}\alpha_{s}ds+\int_{0}^{t}\beta_{s}dw_{s}=\eta_{t}.
\]
Conversely, assume that $\eta_{t}=\mathbb{E}[\xi_{t}|\mathcal{F}_{t}]$
with some absolutely continuous process $\xi.$ Related processes
$\zeta,g_{t}$ are defined in \eqref{abs_cont_process}, \eqref{adapted_derivatives}.
Consider the process 
\[
M_{t}=\eta_{t}-\eta_{0}-\int_{0}^{t}\mathbb{E}[\zeta_{r}|\mathcal{F}_{r}]dr,\ 0\leq t\leq T.
\]
We verify that the process $M$ is a martingale. Indeed, for $s<t$
we compute
\begin{align*}
\mathbb{E}[M_{t}|\mathcal{F}_{s}] & =\mathbb{E}\left[\eta_{t}-\eta_{0}-\int_{0}^{t}\mathbb{E}[\zeta_{r}|\mathcal{F}_{r}]dr|\mathcal{F}_{s}\right]=\\
 & =\mathbb{E}\left[\mathbb{E}[\xi_{t}|\mathcal{F}_{t}]-\eta_{0}-\int_{0}^{t}\mathbb{E}[\zeta_{r}|\mathcal{F}_{r}]dr|\mathcal{F}_{s}\right]=\\
 & \mathbb{E}[\xi_{t}|\mathcal{F}_{s}]-\eta_{0}-\int_{s}^{t}\mathbb{E}[\zeta_{r}|\mathcal{F}_{s}]dr-\int_{0}^{s}\mathbb{E}[\zeta_{r}|\mathcal{F}_{r}]dr=\\
 & =\mathbb{E}\left[\xi_{s}+\int_{s}^{t}\zeta_{r}dr|\mathcal{F}_{s}\right]-\eta_{0}-\int_{s}^{t}\mathbb{E}[\zeta_{r}|\mathcal{F}_{s}]dr-\int_{0}^{s}\mathbb{E}[\zeta_{r}|\mathcal{F}_{r}]dr=\\
 & =\mathbb{E}\left[\xi_{s}|\mathcal{F}_{s}\right]-\eta_{0}-\int_{0}^{s}\mathbb{E}[\zeta_{r}|\mathcal{F}_{r}]dr=\eta_{s}-\eta_{0}-\int_{0}^{s}\mathbb{E}[\zeta_{r}|\mathcal{F}_{r}]dr=M_{s}.
\end{align*}
As a Wiener martingale, the process $M$ has a representation as a
stochastic integral. Since $M_{0}=0$ the representation is of the
form:
\[
M_{t}=\int_{0}^{t}v_{s}dw_{s}
\]
for some progressively measurable square integrable process $v=(v_{s})_{s\in[0,T]}.$
To find this process we use Itô's isometry. For arbitrary progressively
measurable square integrable process $u=(u_{s})_{s\in[0,T]},$ we
have 
\begin{align*}
 & \mathbb{E}\int_{0}^{T}v_{s}u_{s}ds=\mathbb{E}\int_{0}^{T}v_{s}dw_{s}\int_{0}^{T}u_{s}dw_{s}=\mathbb{E}M_{T}\int_{0}^{T}u_{s}dw_{s}=\\
 & =\mathbb{E}\left(\xi_{T}-\eta_{0}-\int_{0}^{T}\mathbb{E}[\zeta_{r}|\mathcal{F}_{r}]dr\right)\int_{0}^{T}u_{s}dw_{s}=\mathbb{E}\left(\xi_{T}\int_{0}^{T}u_{s}dw_{s}\right)-\int_{0}^{T}\mathbb{E}\left(\zeta_{r}\int_{0}^{r}u_{s}dw_{s}\right)dr=\\
 & =\mathbb{E}\int_{0}^{T}g_{T,s}u_{s}ds-\int_{0}^{T}\mathbb{E}\int_{0}^{r}h_{r,s}u_{s}dsdr=\mathbb{E}\int_{0}^{T}\left(g_{T,s}-\int_{s}^{T}h_{r,s}dr\right)u_{s}ds.
\end{align*}
From the previous Lemma we deduce
\[
v_{s}=g_{T,s}-\int_{s}^{T}h_{r,s}dr=g_{0,s}+\int_{0}^{s}h_{r,s}dr=g_{s,s}.
\]
Finally, equality
\[
\eta_{t}=\eta_{0}+\int_{0}^{t}\mathbb{E}[\zeta_{s}|\mathcal{F}_{s}]ds+M_{t}=\eta_{0}+\int_{0}^{t}\mathbb{E}[\zeta_{s}|\mathcal{F}_{s}]ds+\int_{0}^{t}g_{s.s}dw_{s}
\]
implies the needed semimartingale representation of $\eta.$
\end{proof}
\begin{rem*}
For almost all $t\in[0,T]$ the process $\zeta_{t}$ is the time-derivative
of the process $\xi_{t}.$ If we assume that $\xi_{t}$ is Malliavin
differentiable, then the Clark-Ocone formula implies
\[
g_{t,t}=\mathbb{E}\left[D_{t}\xi_{t}|\mathcal{F}_{t}\right].
\]
In these terms the result of the theorem can be written as
\[
d\mathbb{E}\left[\xi_{t}|\mathcal{F}_{t}\right]=\mathbb{E}\left[\partial_{t}\xi_{t}|\mathcal{F}_{t}\right]dt+\mathbb{E}\left[D_{t}\xi_{t}|\mathcal{F}_{t}\right]dw_{t},
\]
which is the motivation behind the term ``total derivative formula''.
\end{rem*}

\subsection{\label{subsec:Characterization-of-It=0000F4}Characterization of
Itô processes: converse of Proposition \ref{prop: Total derivative formula}}

The total derivative formula states that for any square integrable
absolutely continuous process $(\eta_{t})_{t\in[0,T]},$ the process
$\xi_{t}=E[\eta_{t}|\mathcal{F}_{t}]$ is an Itô process. In this
section we give the converse statement.
\begin{claim*}
For any square integrable Itô process $(\xi_{t})_{t\in[0,T]}$ there
exists an absolutely continuous process $(\eta_{t})_{t\in[0,T]}$
such that
\[
\xi_{t}=E[\eta_{t}|\mathcal{F}_{t}].
\]
\end{claim*}
\begin{proof}
By the definition of the Itô process, there exist adapted square integrable
processes $(\alpha_{t})_{t\in[0,T]}$ and $(\beta_{t})_{t\in[0,T]}$
such that 
\[
\xi_{t}=\xi_{0}+\int_{0}^{t}\alpha_{s}ds+\int_{0}^{t}\beta_{s}dw_{s}.
\]
We simply define 
\[
\eta_{t}=\xi_{0}+\int_{0}^{t}\alpha_{s}ds+\int_{0}^{T}\beta_{s}dw_{s}.
\]
The process $\eta$ is absolutely continuous, and 
\[
E[\eta_{t}|\mathcal{F}_{t}]=\xi_{0}+\int_{0}^{t}\alpha_{s}ds+E\left[\int_{0}^{T}\beta_{s}dw_{s}\bigg|\mathcal{F}_{t}\right]=
\]
\[
=\xi_{0}+\int_{0}^{t}\alpha_{s}ds+\int_{0}^{t}\beta_{s}dw_{s}=\xi_{t}.
\]
\end{proof}

\newpage{}

\section*{References}

\[
\]

Acemoglu, Daron, Philippe Aghion, Leonardo Bursztyn, and David Hemous.
\textquotedbl The environment and directed technical change.\textquotedbl{}
American economic review 102, no. 1 (2012): 131-66.

Acemoglu, Daron, Philippe Aghion, Lint Barrage and David Hemous \textquotedbl Climate
change, directed innovation, and energy transition: The long-run consequences
of the shale gas revolution.\textquotedbl{} In 2019 Meeting Papers,
no. 1302. Society for Economic Dynamics, 2019.

Acemoglu, Daron, Ufuk Akcigit, Douglas Hanley, and William Kerr. \textquotedbl Transition
to clean technology.\textquotedbl{} Journal of Political Economy 124,
no. 1 (2016): 52-104.

Acemoglu, Daron, Georgy Egorov, and Konstantin Sonin. Institutional
Change and Institutional Persistence. No. w27852. National Bureau
of Economic Research, 2020.

Aghion, Philippe, Antoine Dechezleprêtre, David Hemous, Ralf Martin,
and John Van Reenen. \textquotedbl Carbon taxes, path dependency,
and directed technical change: Evidence from the auto industry.\textquotedbl{}
Journal of Political Economy 124, no. 1 (2016): 1-51.

Aghion, Philippe, Cameron Hepburn, Alexander Teytelboym, and Dimitri
Zenghelis. \textquotedbl Path dependence, innovation and the economics
of climate change.\textquotedbl{} In Handbook on Green Growth. Edward
Elgar Publishing, 2019.

Akerlof, George A., and Janet L. Yellen. \textquotedbl A near-rational
model of the business cycle, with wage and price inertia.\textquotedbl{}
The Quarterly Journal of Economics 100, no. Supplement (1985): 823-838.

Alvarez, Fernando, and Francesco Lippi. \textquotedbl The Analytic
Theory of a Monetary Shock.\textquotedbl{} Einaudi Institute for Economics
and Finance Mimeo, in preparation (2019).

Alvarez, Fernando E., Francesco Lippi, and Luigi Paciello. \textquotedbl Optimal
price setting with observation and menu costs.\textquotedbl{} The
Quarterly journal of economics 126, no. 4 (2011): 1909-1960.

Alvarez, Fernando E., Francesco Lippi, and Luigi Paciello. \textquotedbl Monetary
shocks in models with inattentive producers.\textquotedbl{} The Review
of economic studies 83, no. 2 (2016): 421-459.

Alvarez, Fernando, Herve Le Bihan, and Francesco Lippi. \textquotedbl The
real effects of monetary shocks in sticky price models: a sufficient
statistic approach.\textquotedbl{} American Economic Review 106, no.
10 (2016): 2817-51.

Alvarez, Fernando E., Francesco Lippi, and Aleksei Oskolkov. The Macroeconomics
of Sticky Prices with Generalized Hazard Functions. No. w27434. National
Bureau of Economic Research, 2020.

Arrow, Kenneth J. \textquotedbl Increasing returns: historiographic
issues and path dependence.\textquotedbl{} The European Journal of
the History of Economic Thought 7, no. 2 (2000): 171-180.

Bakkensen, Laura A., and Lint Barrage. Flood risk belief heterogeneity
and coastal home price dynamics: Going under water?. No. w23854. National
Bureau of Economic Research, 2017.

Baldwin, Elizabeth, Yongyang Cai, and Karlygash Kuralbayeva. \textquotedbl To
build or not to build? Capital stocks and climate policy\textasteriskcentered .\textquotedbl{}
Journal of Environmental Economics and Management 100 (2020): 102235.

Barnett, Michael, William Brock, and Lars Peter Hansen. \textquotedbl Pricing
uncertainty induced by climate change.\textquotedbl{} The Review of
Financial Studies 33, no. 3 (2020): 1024-1066.

Barrage, Lint, and Jacob Furst. \textquotedbl Housing investment,
sea level rise, and climate change beliefs.\textquotedbl{} Economics
letters 177 (2019): 105-108.

Bergemann, Dirk, and Philipp Strack. \textquotedbl Dynamic revenue
maximization: A continuous time approach.\textquotedbl{} Journal of
Economic Theory 159 (2015): 819-853.

Berger, David W., Konstantin Milbradt, Fabrice Tourre, and Joseph
Vavra. Mortgage prepayment and path-dependent effects of monetary
policy. No. w25157. National Bureau of Economic Research, 2018.

Blanchard, Olivier J., and Lawrence H. Summers. \textquotedbl Hysteresis
and the European unemployment problem.\textquotedbl{} NBER macroeconomics
annual 1 (1986): 15-78.Bhamra, Harjoat S., and Raman Uppal. \textquotedbl The
effect of introducing a non-redundant derivative on the volatility
of stock-market returns when agents differ in risk aversion.\textquotedbl{}
The Review of Financial Studies 22, no. 6 (2009): 2303-2330.

Borovi\v{c}ka, J. and Hansen, L.P., 2016. Term structure of uncertainty
in the macroeconomy. In Handbook of Macroeconomics (Vol. 2, pp. 1641-1696).
Elsevier.

Borovi\v{c}ka, Jaroslav, Lars Peter Hansen, and José A. Scheinkman.
\textquotedbl Shock elasticities and impulse responses.\textquotedbl{}
Mathematics and Financial Economics 8, no. 4 (2014): 333-354.

Borovi\v{c}ka, Jaroslav, Lars Peter Hansen, Mark Hendricks, and José
A. Scheinkman. \textquotedbl Risk-price dynamics.\textquotedbl{}
Journal of Financial Econometrics 9, no. 1 (2011): 3-65.

Boucher, O., P. R. Halloran, E. J. Burke, M. Doutriaux-Boucher, C.
D. Jones, J. Lowe, M. A. Ringer, E. Robertson, and P. Wu. \textquotedbl Reversibility
in an Earth System model in response to CO2 concentration changes.\textquotedbl{}
Environmental Research Letters 7, no. 2 (2012): 024013.

Boulatov, Alexei, Georgii Riabov, and Aleh Tsyvinski. ``Optimal trading
with dynamic stochastic price impact'', Working paper (2020).

Brian, Arthur, W. \textquotedbl Competing technologies, increasing
returns, and lock-in by historical events.\textquotedbl{} The economic
journal 99, no. 394 (1989): 116-131.

Brock, William A., and Lars Peter Hansen. \textquotedbl Wrestling
with uncertainty in climate economic models.\textquotedbl{} University
of Chicago, Becker Friedman Institute for Economics Working Paper
2019-71 (2018).

Cai, Yongyang, and Thomas S. Lontzek. \textquotedbl The social cost
of carbon with economic and climate risks.\textquotedbl{} Journal
of Political Economy 127, no. 6 (2019): 2684-2734.

Collins, Matthew, Reto Knutti, Julie Arblaster, Jean-Louis Dufresne,
Thierry Fichefet, Pierre Friedlingstein, Xuejie Gao et al. \textquotedbl Long-term
climate change: projections, commitments and irreversibility.\textquotedbl{}
In Climate Change 2013-The Physical Science Basis: Contribution of
Working Group I to the Fifth Assessment Report of the Intergovernmental
Panel on Climate Change, pp. 1029-1136. Cambridge University Press,
2013.

Cont, Rama, and David-Antoine Fournie. \textquotedbl Functional Itô
calculus and stochastic integral representation of martingales.\textquotedbl{}
The Annals of Probability 41, no. 1 (2013): 109-133.

Cosso, Andrea, Salvatore Federico, Fausto Gozzi, Mauro Rosestolato,
and Nizar Touzi. \textquotedbl Path-dependent equations and viscosity
solutions in infinite dimension.\textquotedbl{} The Annals of Probability
46, no. 1 (2018): 126-174.

Cvitani\'{c}, Jaksa, and Semyon Malamud. \textquotedbl Equilibrium
driven by discounted dividend volatility.\textquotedbl{} Swiss Finance
Institute Research Paper 09-34 (2009)

Cvitani\'{c}, Jakša, Dylan Possamaï, and Nizar Touzi. \textquotedbl Moral
hazard in dynamic risk management.\textquotedbl{} Management Science
63, no. 10 (2017): 3328-3346.

David, Paul A. \textquotedbl Clio and the Economics of QWERTY.\textquotedbl{}
The American Economic Review 75, no. 2 (1985): 332-337.

Detemple, Jerome B., and Fernando Zapatero. \textquotedbl Asset prices
in an exchange economy with habit formation.\textquotedbl{} Econometrica:
Journal of the Econometric Society (1991): 1633-1657.

Dietz, Simon, James Rising, Thomas Stoerk, and Gernot Wagner. \textquotedbl Tipping
Points in the Climate System and the Economics of Climate Change.\textquotedbl{}
In EGU General Assembly Conference Abstracts, p. 2959. 2020.

Dixit, Avinash. \textquotedbl Analytical approximations in models
of hysteresis.\textquotedbl{} The Review of Economic Studies 58, no.
1 (1991): 141-151.

Dixit, Avinash. \textquotedbl Investment and hysteresis.\textquotedbl{}
Journal of economic perspectives 6, no. 1 (1992): 107-132.

Dupire, Bruno. \textquotedbl Functional Itô calculus.\textquotedbl{}
2009. Bloomberg Working Paper.

Dupire, Bruno. \textquotedbl Functional Itô calculus.\textquotedbl{}
Quantitative Finance 19, no. 5 (2019): 721-729. 

Eggertsson, Gauti B., Neil R. Mehrotra, and Jacob A. Robbins. \textquotedbl A
model of secular stagnation: Theory and quantitative evaluation.\textquotedbl{}
American Economic Journal: Macroeconomics 11, no. 1 (2019): 1-48.

Egorov, Georgy, and Konstantin Sonin. The Political Economics of Non-democracy.
No. w27949. National Bureau of Economic Research, 2020.

Eliseev, Alexey V., Pavel F. Demchenko, Maxim M. Arzhanov, and Igor
I. Mokhov. \textquotedbl Transient hysteresis of near-surface permafrost
response to external forcing.\textquotedbl{} Climate dynamics 42,
no. 5-6 (2014): 1203-1215.

Fabbri, Giorgio, Fausto Gozzi, and Andrzej Swiech. \textquotedbl Stochastic
optimal control in infinite dimension.\textquotedbl{} Probability
and Stochastic Modelling. Springer (2017).

Fouquet, Roger. \textquotedbl Path dependence in energy systems and
economic development.\textquotedbl{} Nature Energy 1, no. 8 (2016):
1-5.

Fournié, Eric, Jean-Michel Lasry, Jérôme Lebuchoux, Pierre-Louis Lions,
and Nizar Touzi. \textquotedbl Applications of Malliavin calculus
to Monte Carlo methods in finance.\textquotedbl{} Finance and Stochastics
3, no. 4 (1999): 391-412.

Galí, Jordi. Insider-outsider labor markets, hysteresis and monetary
policy. No. w27385. National Bureau of Economic Research, 2020.

Garbe, Julius, Torsten Albrecht, Anders Levermann, Jonathan F. Donges,
and Ricarda Winkelmann. \textquotedbl The hysteresis of the Antarctic
ice sheet.\textquotedbl{} Nature 585, no. 7826 (2020): 538-544.

Gasser, Thomas, Mehdi Kechiar, Phillippe Ciais, E. J. Burke, Thomas
Kleinen, Dan Zhu, Yuanyuan Huang, Altug Ekici, and M. Obersteiner.
\textquotedbl Path-dependent reductions in CO 2 emission budgets
caused by permafrost carbon release.\textquotedbl{} Nature Geoscience
11, no. 11 (2018): 830-835.

Giglio, Stefano, Bryan T. Kelly, and Johannes Stroebel. Climate Finance.
No. w28226. National Bureau of Economic Research, 2020.

Golosov, Mikhail, John Hassler, Per Krusell, and Aleh Tsyvinski. \textquotedbl Optimal
taxes on fossil fuel in general equilibrium.\textquotedbl{} Econometrica
82, no. 1 (2014): 41-88.

Grubb, Michael, Rutger-Jan Lange, Nicolas Cerkez, Pablo Salas, Jean-Francois
Mercure, and Ida Sognnaes. \textquotedbl Taking Time Seriously: Implications
for Optimal Climate Policy.\textquotedbl{} (2020).

Hassler, John, Per Krusell, Conny Olovsson, and Michael Reiter. \textquotedbl On
the effectiveness of climate policies.\textquotedbl{} IIES WP (2020).

Huang, Zangbo and Michael Sockin ``Dynamic Taxation with Endogenous
Skill Premia'', 2018.

Hong, Harrison, Neng Wang, and Jinqiang Yang. Mitigating Disaster
Risks in the Age of Climate Change. National Bureau of Economic Research,
2020.

Jordà, Òscar, Sanjay R. Singh, and Alan M. Taylor. The long-run effects
of monetary policy. No. w26666. National Bureau of Economic Research,
2020.

Kotlikoff, L.J., Kubler, F., Polbin, A. and Scheidegger, S., 2020.
Pareto-Improving Carbon-Risk Taxation (No. w26919). National Bureau
of Economic Research.

Krasnosel'skii, Mark A., and Aleksei V. Pokrovskii. Systems with hysteresis.
Springer Science \& Business Media, 2012.

Lemoine, Derek. \textquotedbl The climate risk premium: how uncertainty
affects the social cost of carbon.\textquotedbl{} Journal of the Association
of Environmental and Resource Economists 8, no. 1 (2021): 27-57.

Lemoine, Derek, and Christian Traeger. \textquotedbl Watch your step:
optimal policy in a tipping climate.\textquotedbl{} American Economic
Journal: Economic Policy 6, no. 1 (2014): 137-66.

Li, Xin, Borghan Narajabad, and Ted Temzelides. \textquotedbl Robust
dynamic energy use and climate change.\textquotedbl{} Quantitative
Economics 7, no. 3 (2016): 821-857.

Ljungqvist, Lars, and Thomas J. Sargent. Recursive macroeconomic theory.
MIT press, 2018.

Lontzek, Thomas S., Yongyang Cai, Kenneth L. Judd, and Timothy M.
Lenton. \textquotedbl Stochastic integrated assessment of climate
tipping points indicates the need for strict climate policy.\textquotedbl{}
Nature Climate Change 5, no. 5 (2015): 441-444.

Mankiw, N. Gregory. \textquotedbl Small menu costs and large business
cycles: A macroeconomic model of monopoly.\textquotedbl{} The Quarterly
Journal of Economics 100, no. 2 (1985): 529-537.

Makris, Miltiadis, and Alessandro Pavan. Taxation under learning-by-doing.
Working Paper, 2017.

Makris, Miltiadis, and Alessandro Pavan. Wedge Dynamics with Endogenous
Private Information: A General Recursive Characterization. mimeo,
Northwestern University, 2018.

Meng, Kyle C. Estimating path dependence in energy transitions. No.
w22536. National Bureau of Economic Research, 2016.

Nohara, Daisuke, Yoshikatsu Yoshida, Kazuhiro Misumi, and Masamichi
Ohba. \textquotedbl Dependency of climate change and carbon cycle
on CO2 emission pathways.\textquotedbl{} Environmental Research Letters
8, no. 1 (2013): 014047.

Nordhaus, W.D., 1994. Managing the global commons: the economics of
climate change (Vol. 31). Cambridge, MA: MIT press.

Page, Scott E. \textquotedbl Path dependence.\textquotedbl{} Quarterly
Journal of Political Science 1, no. 1 (2006): 87-115.

Pavan, Alessandro. \textquotedbl Dynamic mechanism design: Robustness
and endogenous types.\textquotedbl{} In Advances in Economics and
Econometrics: Eleventh World Congress, pp. 1-62. 2017.

Pavan, Alessandro, Ilya Segal, and Juuso Toikka. \textquotedbl Dynamic
mechanism design: A myersonian approach.\textquotedbl{} Econometrica
82, no. 2 (2014): 601-653.

Reis, Ricardo. \textquotedbl Inattentive producers.\textquotedbl{}
The Review of Economic Studies 73, no. 3 (2006): 793-821.

Prescott, Edward C. \textquotedbl Should control theory be used for
economic stabilization?.\textquotedbl{} In Carnegie-Rochester Conference
Series on Public Policy, vol. 7, pp. 13-38. North-Holland, 1977.

Rogelj, J., D. Shindell, K. Jiang, S. Fifita, P. Forster, V. Ginzburg,
C. Handa, H. Kheshgi, S. Kobayashi, E. Kriegler, L. Mundaca, R. Séférian,
and M.V.Vilariño, 2018: Mitigation Pathways Compatible with 1.5°C
in the Context of Sustainable Development. In: Global Warming of 1.5°C.
An IPCC Special Report on the impacts of global warming of 1.5°C above
pre-industrial levels and related global greenhouse gas emission pathways,
in the context of strengthening the global response to the threat
of climate change, sustainable development, and efforts to eradicate
poverty {[}Masson-Delmotte, V., P. Zhai, H.-O. Pörtner, D. Roberts,
J. Skea, P.R. Shukla, A. Pirani, W. Moufouma-Okia, C. Péan, R. Pidcock,
S. Connors, J.B.R. Matthews, Y. Chen, X. Zhou, M.I. Gomis, E. Lonnoy,
T. Maycock, M. Tignor, and T. Waterfield (eds.){]}. In Press.

Sannikov, Yuliy. \textquotedbl Moral hazard and long-run incentives.\textquotedbl{}
Unpublished working paper, Princeton University (2014).

Serrat, Angel. \textquotedbl A dynamic equilibrium model of international
portfolio holdings.\textquotedbl{} Econometrica 69, no. 6 (2001):
1467-1489.

Summers, Lawrence H. \textquotedbl US economic prospects: Secular
stagnation, hysteresis, and the zero lower bound.\textquotedbl{} Business
economics 49, no. 2 (2014): 65-73.

Temzelides, Ted. \textquotedbl Needed: Robustness in Climate Economics.\textquotedbl{}
Coping with the climate crisis (2016): 104.

Traeger, Christian P. \textquotedbl Ace--analytic climate economy
(with temperature and uncertainty).\textquotedbl{} (2018).

Van Den Bremer, Ton, and Rick van der Ploeg. \textquotedbl Pricing
Carbon Under Economic and Climactic Risks: Leading-Order Results from
Asymptotic Analysis.\textquotedbl{} (2018).

van der Ploeg, Frederick, and Aart de Zeeuw. \textquotedbl Climate
tipping and economic growth: precautionary capital and the price of
carbon.\textquotedbl{} Journal of the European Economic Association
16, no. 5 (2018): 1577-1617.

Walsh, Conor. \textquotedbl Hysteresis Implies Scale Effects.\textquotedbl{}
(2020): Working paper.

Zickfeld, K., Arora, V.K. and Gillett, N.P., 2012. Is the climate
response to CO2 emissions path-dependent?. Geophysical Research Letters,
39(5).

\end{document}